\documentclass[amsart,12pt,oneside,english]{article}
\usepackage{amsthm, amsmath, amsfonts, amsxtra, amssymb, euscript, mathrsfs, MnSymbol, verbatim, enumerate, multicol, multirow, color,babel, geometry, tikz, tikz-cd, tikz-3dplot, tkz-graph, array, enumitem, hyperref, thm-restate, thmtools, datetime, graphicx, tensor, braket, slashed, adjustbox, mathtools, etoolbox, bbding, esint, pgfplots, ytableau, float, rotating, mathdots, savesym, cite, wasysym, amscd, pifont, setspace, cleveref, wrapfig, picture, tabularx, youngtab, parskip, subcaption,caption}
\usepackage[numbers,sort&compress]{natbib}
\usepackage[normalem]{ulem} 
\usepackage[utf8]{inputenc}
\usepackage[all]{xy}
\numberwithin{equation}{section}
\numberwithin{figure}{section}
\usetikzlibrary{arrows, positioning, decorations.pathmorphing, decorations.markings, decorations.pathreplacing, decorations.markings, matrix, patterns, snakes}
\tikzset{
  on each segment/.style={
    decorate,
    decoration={
      show path construction,
      moveto code={},
      lineto code={
        \path [#1]
        (\tikzinputsegmentfirst) -- (\tikzinputsegmentlast);
      },
      curveto code={
        \path [#1] (\tikzinputsegmentfirst)
        .. controls
        (\tikzinputsegmentsupporta) and (\tikzinputsegmentsupportb)
        ..
        (\tikzinputsegmentlast);
      },
      closepath code={
        \path [#1]
        (\tikzinputsegmentfirst) -- (\tikzinputsegmentlast);
      },
    },
  },
  mid arrow/.style={postaction={decorate,decoration={
        markings,
        mark=at position .7 with {\arrow[#1]{stealth}}
      }}},
}
\hypersetup{colorlinks=true}
\hypersetup{linkcolor=black}
\hypersetup{citecolor=black}
\hypersetup{urlcolor=black}

\makeatletter
\def\oversortoftilde#1{\mathop{\vbox{\m@th\ialign{##\crcr\noalign{\kern3\p@}%
				\sortoftildefill\crcr\noalign{\kern3\p@\nointerlineskip}%
				$\hfil\displaystyle{#1}\hfil$\crcr}}}\limits}

\def\sortoftildefill{$\m@th \setbox\z@\hbox{$-$}%
	\braceld\leaders\vrule \@height\ht\z@ \@depth\z@\hfill\braceru$}

\makeatother

\theoremstyle{plain}
\newtheorem*{thm*}{Theorem}
\newtheorem{thm}{Theorem}[section]

\newtheorem{lem}[thm]{Lemma}

\theoremstyle{definition}
\newtheorem{defn}[thm]{Definition}
\newtheorem*{defn*}{Definition}

\newtheorem*{thm:maintheoremCstar}{Theorem \ref{thm:maintheoremCstar}}

\crefname{lemma}{lemma}{lemmas}
\Crefname{lemma}{Lemma}{Lemmas}
\crefname{thm}{theorem}{theorems}
\Crefname{thm}{Theorem}{Theorems}
\crefname{defn}{definition}{definitions}
\Crefname{defn}{Definition}{Definitions}

\DeclarePairedDelimiterX{\abs}[1]{\lvert}{\rvert}{\ifblank{#1}{{}\cdot{}}{#1}}

\makeatother
\newcommand{\calo}{\mathcal{O}}

\newcommand{\calp}{\mathcal{P}}
\newcommand{\cals}{\mathcal{S}}
\newcommand{\calh}{\mathcal{H}}

\newcommand{\calb}{\mathcal{B}}

\begin{document}

\begin{titlepage}
\vspace*{-3cm} 
\begin{flushright}
{\tt CALT-TH-2020-020}\\
\end{flushright}
\begin{center}
\vspace{2.5cm}
{\Large\bfseries Thermal states are vital: Entanglement Wedge Reconstruction from Operator-Pushing \\  }
\vspace{2cm}
{\large
Elliott Gesteau$^{1,2}$ and Monica Jinwoo Kang$^3$\\}
\vspace{.6cm}
{ $^1$ Département de Mathématiques et Applications, Ecole Normale Supérieure}\par\vspace{-.3cm}
{Paris, 75005, France}\par
{ $^2$ Perimeter Institute for Theoretical Physics}\par\vspace{-.3cm}
{Waterloo, Ontario N2L 2Y5, Canada}\par
\vspace{.2cm}
{ $^3$ Walter Burke Institute for Theoretical Physics, California Institute of Technology}\par\vspace{-.3cm}
{  Pasadena, CA 91125, U.S.A.}\par
\vspace{.6cm}

\scalebox{.95}{\tt  egesteau@perimeterinstitute.ca, monica@caltech.edu}\par
\vspace{2cm}
{\bf{Abstract}}\\
\end{center}
{We give a general construction of a setup that verifies bulk reconstruction, conservation of relative entropies, and equality of modular flows between the bulk and the boundary, for infinite-dimensional systems with operator-pushing. In our setup, a bulk-to-boundary map is defined at the level of the $C^*$-algebras of state-independent observables. We then show that if the boundary dynamics allow for the existence of a KMS state, physically relevant Hilbert spaces and von Neumann algebras can be constructed directly from our framework. Our construction should be seen as a state-dependent construction of the other side of a wormhole and clarifies the meaning of black hole reconstruction claims such as the Papadodimas-Raju proposal. As an illustration, we apply our result to construct a wormhole based on the HaPPY code, which satisfies all properties of entanglement wedge reconstruction.}
\\
\vfill 
\end{titlepage}

\tableofcontents
\newpage
\section{Introduction}

The AdS/CFT correspondence \cite{Maldacena:1997re} relates type IIB superstring theory in the bulk of an asymptotically AdS spacetime to a conformal field theory on its boundary, and is probably our best understood theory of quantum gravity. In this context, the bulk and boundary theories are both expected to be described in terms of the algebras of local observables of the CFT. More precisely, local operators of the emerging $d$-dimensional bulk theory can be expressed as operators of a $(d-1)$-dimensional boundary CFT smeared over the entire spatial slice or compact spatial subregions \cite{Hamilton:2006az,Hamilton2005}. Much progress has been made in recent years on understanding the semiclassical limit of AdS/CFT with the framework of Quantum Error Correction \cite{Almheiri:2014lwa}. Indeed, local bulk operators on a fixed geometry can be represented on the boundary in many different ways, and this redundancy makes it natural to view the space of effective field theory states on a fixed geometry in the bulk as a code subspace of the physical Hilbert space of states.

Quantum Error Correction has shed light on deep connections between bulk reconstruction, (relative) entanglement entropy, which can be viewed as an information-theoretic quantity in the scope of quantum error correction, and the bulk geometry via the Ryu-Takayanagi formula \cite{RT2006,Jafferis:2015del,albion,Harlow:2018fse,DongHarlowWall}. For finite-dimensional Hilbert spaces, a rigorous synthetic statement was proven in \cite{Harlow:2016vwg}, and establishes the equivalence between bulk reconstruction, the Ryu-Takayanagi formula, and the equivalence between bulk and boundary relative entropies. This statement is very well-suited to describe some finite-dimensional toy models of AdS/CFT like finite tensor networks. In particular, the HaPPY code has proven to be a very successful tensor network model of holographic quantum field theories \cite{Pastawski:2015qua}.

However, in a more realistic setting, we expect the bulk and boundary Hilbert spaces to be infinite-dimensional.\footnote{This claim can be justified in several ways. First, the Reeh-Schlieder theorem, which is at the foundation of Quantum Field Theory, can only be satisfied in infinite dimensions. Furthermore, some important physical phenomena, like spontaneous symmetry breaking, can only be accounted for in infinite dimensions. Finally, a very recent breakthrough \cite{Ji:2019} shows that not all quantum correlations can be reproduced by finite-dimensional systems.} 
In \cite{Kang:2018xqy,Kang:2019dfi}, the connection between entanglement wedge reconstruction and relative entropy equivalence between the bulk and the boundary was extended to infinite-dimensional Hilbert spaces with (infinite-dimensional) von Neumann algebras as operator algebras acting on them. Under conditions that amount to a bulk equivalent of the Reeh--Schlieder theorem, \cite{Kang:2018xqy} showed the equivalence between bulk reconstruction and the equality of bulk and boundary relative entropies. The Ryu-Takayanagi formula stays out of reach, as it must rely on regulating schemes in the infinite-dimensional case. This result, which we will study in detail, works well for simple toy models such as the one studied in \cite{Kang:2019dfi}, but has the disadvantage of being formulated with a bulk reconstruction at the level of the \textit{states}, rather than the \textit{operators}. This formulation is not enough to study more realistic tensor network models relevant to holography (such as the HaPPY code), which are much better-suited to {\em operator pushing} than {\em state pushing}. {\em Operator pushing} means that operators acting on the bulk algebra are directly mapped to boundary operators, whereas {\em state pushing} would map bulk states to boundary states. If one wants to consider infinite-dimensional counterparts of such codes, it seems necessary to formulate an infinite-dimensional statement directly at the level of the operators.

Mathematically, this means that instead of looking at von Neumann algebras of operators acting on explicit Hilbert spaces, one will instead approach with a broader notion of $C^*$-algebras, which allows one to formulate bulk reconstruction directly without any reference to a Hilbert space.\footnote{It is possible to define von Neumann algebras without the specification of Hilbert spaces. In this case, where von Neumann algebras are referred to as $W^*$-algebras, they are defined as $C^*$-algebras which possess a \textit{predual}. However, this characterization is equivalent to asking that a $C^*$-algebra already defined on a Hilbert space is its own bicommutant, and therefore captures more of the Hilbert space structure.} 
In this paper, we prove that if bulk reconstruction is satisfied at the level of the operators while the dynamics of a physical system satisfy some conditions relevant to a physical setting, then it is possible to construct physically relevant Hilbert spaces with a bulk-to-boundary mapping that satisfies bulk reconstruction out of thermal vacua.\footnote{The notion of thermality will be encoded in the KMS condition throughout the paper. A definition of the KMS condition is given in Section \ref{sec:AQFT}.} This theorem will be referred to as Theorem \ref{thm:maintheoremCstar}.
\bigskip
\begin{thm}
\label{thm:maintheoremCstar}
Let $\mathcal{A}_{code}$ and $\mathcal{A}_{phys}$ be two $C^*$-algebras, and let $\iota:\mathcal{A}_{code}\longrightarrow \mathcal{A}_{phys}$ be an isometric $C^*$-homomorphism. Let $\sigma_t$ be a strongly continuous one-parameter group of isometries of $\mathcal{A}_{phys}$ such that $\sigma_t(\iota(\mathcal{A}_{code}))\subset\iota(\mathcal{A}_{code})$, and $\omega$ be a KMS state on $\mathcal{A}_{phys}$ with respect to $\sigma_t$ at inverse temperature $\beta$. Then there exist Hilbert space representations $(\pi_\omega^{phys}, \mathcal{H}_{phys})$ and $(\pi_\omega^{code}, \mathcal{H}_{code})$ of $\mathcal{A}_{phys}$ and $\mathcal{A}_{code}$ such that:
\begin{enumerate}
\item there exists a Hilbert space isometry $u:\mathcal{H}_{code}\longrightarrow\mathcal{H}_{phys}$ such that
$$\forall A\in\mathcal{A}_{code},\quad \pi_\omega^{phys}(\iota(A))u=u\pi_\omega^{code}(A).$$

\item there exists a vector $\ket{\Omega}_{code}\in\mathcal{H}_{code}$ and a vector $\ket{\Omega}_{phys}\in\mathcal{H}_{phys}$ such that 
\begin{align*}
\forall A\in\mathcal{A}_{phys},\quad\omega(A)&=\bra{\Omega_{phys}}\pi_\omega^{phys}(A)\ket{\Omega_{phys}},\\
\forall A\in\mathcal{A}_{code},\quad\omega(\iota(A))&=\bra{\Omega_{code}}\pi_\omega^{code}(A)\ket{\Omega_{code}}.
\end{align*}

\item if $M_{code}=\pi_{\omega}^{code}(\mathcal{A}_{code})''$ and $M_{phys}=\pi_{\omega}^{phys}(\mathcal{A}_{phys})''$\footnote{We denote a double commutant of $M$ by $M''$.}, then $\ket{\Omega_{code}}$ is cyclic and separating with respect to $M_{code}$ and $\ket{\Omega_{phys}}$ is cyclic and separating with respect to $M_{phys}$. Moreover,\begin{align} \nonumber
\begin{split}
\forall \calo \in M_{code}\ \forall \calo^\prime \in M_{code}^\prime, \quad 
\exists\tilde{\calo} \in M_{phys}\ \exists \tilde{\calo}^\prime \in M_{phys}^\prime\quad \text{such that}\quad\\
\forall \ket{\Theta} \in \calh_{code} \quad 
\begin{cases}
u \calo \ket{\Theta} =  \tilde{\calo} u \ket{\Theta}, \quad
&u \calo^\prime \ket{\Theta} =  \tilde{\calo}^\prime u \ket{\Theta}, \\
u \calo^\dagger \ket{\Theta} =  \tilde{\calo}^\dagger u \ket{\Theta}, \quad
&u \calo^{\prime \dagger} \ket{\Theta} = \tilde{\calo}^{\prime\dagger} u\ket{\Theta}.
\end{cases}\quad
\end{split}
\end{align}

\item if $\ket{\Phi}$ and $\ket{\Psi}$ are two vectors in $\mathcal{H}_{code}$ with $\ket{\Psi}$ cyclic and separating with respect to $M_{code}$, then $u\ket{\Psi}$ is cyclic and separating with respect to $M_{phys}$ and the equality of the relative entropy holds:
$$\mathcal{S}_{\Psi|\Phi}(M_{code})=\mathcal{S}_{u\Psi|u\Phi}(M_{phys}).$$

\item if $\Delta_{\Omega_{phys}}$ is the modular operator of $\ket{\Omega}_{phys}$ with respect to $M_{phys}$ and $\Delta_{\Omega_{code}}$ is the modular operator of $\ket{\Omega}_{code}$ with respect to $M_{code}$, then 
\begin{align*}
\forall A\in\mathcal{A}_{phys},\quad\pi_{\omega}^{phys}(\sigma_t(A))&=\Delta_{\Omega_{phys}}^{-\frac{it}{\beta}}\pi_{\omega}^{phys}(A)\Delta_{\Omega_{phys}}^\frac{it}{\beta},\\
\forall A\in\mathcal{A}_{code},\quad\pi_{\omega}^{phys}(\sigma_t(\iota(A)))u&=u\Delta_{\Omega_{code}}^{-\frac{it}{\beta}}\pi_{\omega}^{code}(A)\Delta_{\Omega_{code}}^{\frac{it}{\beta}}.
\end{align*}
\end{enumerate}
\end{thm}

The idea of the proof of Theorem \ref{thm:maintheoremCstar} is to show that our conditions satisfy the hypotheses of the theorem in \cite{Kang:2018xqy}, thereby guaranteeing bulk reconstruction at the Hilbert space level, and allowing us to conclude that relative entropies within these Hilbert spaces will be conserved. We also prove that modular evolution will also be conserved. 

Our construction crucially depends on the existence of expectation value functionals on the boundary algebra satisfying the KMS condition, which encodes the notion of thermal equilibrium. We use these functionals as the building blocks of our Hilbert spaces, whose elements should be seen as excitations of a thermal bath. We utilize the Gelfand--Naimark--Segal (GNS) representation to build the Hilbert spaces; in the case of thermal states in AdS/CFT, the GNS representation reduces to the thermofield double (TFD) construction. Therefore, one could see our construction as a general machinery to construct the other side of an AdS/CFT wormhole for a system with operator pushing. 

As expected in the full quantum gravity regime, this construction is state-dependent, in the sense that it entirely relies on our choice of a thermal expectation value functional. In fact, it has clear links with the Papadodimas-Raju proposal in \cite{Papadodimas:2013wnh,Papadodimas:2013jku,Papadodimas:2012aq}: the commutant of the von Neumann algebra of boundary observables is calculated in the same way, and corresponds to the von Neumann algebra of observables on the other boundary.

Theorem \ref{thm:maintheoremCstar} is introduced with holographic tensor networks such as the HaPPY code in mind. Therefore, we apply our construction to the HaPPY code as an example and show that the infinite dimensional HaPPY code is compatible with entanglement wedge reconstruction and relative entropy equivalence between the bulk and the boundary. We summarize the dictionary between our abstract theorem and the HaPPY code in Table \ref{tb:OAandHaPPY}.
\begin{table}[H]
\begin{center}
\arraycolsep=3pt\def\arraystretch{1.4}
\begin{tabular}{|c|c|}
\hline
Operator Algebras & Tensor Networks (cf. HaPPY code)\\
\hline
\hline
bulk  $C^*$-algebra $(\mathcal{A}_{code})$ & local operators on the bulk nodes \\
\hline
boundary  $C^*$-algebra $(\mathcal{A}_{phys})$ & local operators on the boundary nodes \\
\hline
$C^*$ isometry $(\iota)$ & operator pushing map \\
\hline
\begin{tabular}{c}
a strongly-continuous 1-parameter group \\[-7pt]
of isometries of $\mathcal{A}_{phys}$
\end{tabular}$(\sigma_t)$ & trapeze Hamiltonian evolution \\
\hline
a KMS state $(\omega)$ & a thermal vacuum \\
\hline
commutant of the bulk vN algebra $(\mathcal{M}^\prime_{code})$ & the other side of the wormhole\\
\hline
commutant of the boundary vN algebra $(\mathcal{M}^\prime_{phys})$ & the other boundary \\
\hline
\end{tabular}
\end{center}
\caption{The dictionary between the operator algebra, used in the formalism constructed with Theorem \ref{thm:maintheoremCstar}, and holographic tensor networks such as the infinite-dimensional analog of the HaPPY code.}
\label{tb:OAandHaPPY}
\end{table}

An outline of our proof of Theorem \ref{thm:maintheoremCstar} is the following. 
\begin{itemize}
\item We construct the GNS representations of $\mathcal{A}_{code}$ and $\mathcal{A}_{phys}$ with respect to a KMS (thermal equilibrium) state, and prove that the map $\iota$ induces a Hilbert space isometry between them. This provides the first and second point of Theorem \ref{thm:maintheoremCstar}.
\item We carefully check that the hypotheses of the theorem in \cite{Kang:2018xqy} are verified by the induced von Neumann algebras and their commutants in order to prove the third and fourth points of Theorem \ref{thm:maintheoremCstar}. In particular, we use Kaplansky's density theorem \cite{Kaplansky} and the Banach-Alaoglu theorem \cite{Banach,Alaoglu} to show bulk reconstruction, and Tomita--Takesaki theory to extend it to the commutant von Neumann algebras. We also use the fact that KMS states are cyclic and separating for both $C^*$-algebras and von Neumann algebras in their GNS representations.
\item We use Kaplansky's density theorem and the uniqueness part of the Tomita--Takesaki theorem \cite{Tomita-Takesaki} to prove that the modular flow of the vector representative of our KMS state coincides with the extension of the initial $C^*$-algebra flow to the whole von Neumann algebra. This allows us to prove the fifth and last point of Theorem \ref{thm:maintheoremCstar}.
\end{itemize}

The rest of this paper is organized as follows. In Section \ref{sec:Hilbert}, we highlight the differences between $C^*$-algebras and von Neumann algebras and discuss what they mean in the physical context in terms of state-dependence. Then, in Section \ref{sec:AQFT}, we introduce necessary notions relevant to Theorem \ref{thm:maintheoremCstar} borrowed from algebraic quantum field theory, with a particular emphasis on state (in)dependence. In Section \ref{sec:EWreconstruction}, we summarize the proven results on infinite-dimensional entanglement wedge reconstruction in \cite{Kang:2018xqy} and discuss their applicability in our context. With this in mind, we then state and prove Theorem \ref{thm:maintheoremCstar} in Section \ref{sec:main} and discuss its physical implications. For the special cases when the Hilbert spaces are separable, we discuss its simplifications in Section \ref{sec:separableH} and the validity of using the Theorem in \cite{Kang:2018xqy} in our more general context. We exhibit the intricate link between our findings and both the thermofield double construction and the Papadodimas-Raju proposal in Section \ref{sec:TFD}, and highlight the importance of the state-dependence. We illustrate the physical significance of our construction by using our theorem to show that bulk relative entropy equals boundary relative entropy in an infinite-dimensional wormhole analogue of the HaPPY code in Section \ref{sec:HaPPY}. We discuss the physical implications of our result in Section \ref{sec:discussion} and give possible applications, including links with the firewall paradox, superselection theory and a potentially more state-independent way of defining relative entropy.

\section{Analysis in infinite dimensions} \label{sec:Hilbert}

Our setup will make a heavy use of infinite-dimensional operator theory. Here, we provide a lightning review of the essential operator-algebraic concepts underlying our constructions, such as topologies on functional spaces, $C^*$-algebras, and von Neumann algebras.

\subsection{Topolgies on $\mathcal{B}(\mathcal{H})$}
Let $\mathcal{B}(\mathcal{H})$ denote the algebra of bounded operators on $\mathcal{H}$. Then, a topology on $\mathcal{B}(\mathcal{H})$ is a family of subsets of $\calb(\calh)$, which are by definition open. This family must contain both the empty set $\emptyset$ and $\mathcal{B}(\mathcal{H})$ itself, and be stable under finite intersections and arbitrary unions. There exist various notions for a topology on $\mathcal{B}(\mathcal{H})$ and we list the ones necessary for the purpose of this paper in \Cref{def:NormTopo,def:SOT,def:WOT} following the notation of Jones \cite{Jones-vNalg} closely. In this section, $\calo$ is an operator in $\calb(\calh)$ and $\ket{\xi_i}, \ket{\eta_i}$ are states in $\calh$.

\bigskip
\begin{defn} \label{def:NormTopo}
The \emph{norm (or uniform) topology} is induced by the operator norm $||\calo||$. It is the smallest topology that contains the following basic neighborhoods
\begin{equation*} 
\mathcal{N}(\calo,\epsilon) = \{\calp \in \calb(\calh) : ||\calp - \calo|| < \epsilon  \} .
\end{equation*}
\end{defn}
\bigskip
\begin{defn} \label{def:SOT}
The \emph{strong operator topology} is the smallest topology that contains the following basic neighborhoods
\begin{equation*}
\mathcal{N}(\calo,\ket{\xi_1},\ket{\xi_2},\ldots,\ket{\xi_n},\epsilon) = \{\calp \in \calb(\calh) : ||(\calp - \calo)\ket{\xi_i}|| < \epsilon \quad \forall i \in \{ 1,2,\cdots,n\} \} .
\end{equation*}
\end{defn}
The strong operator topology is the weakest topology on $\calb(\calh)$ mapping $E_\psi : \calb(\calh) \rightarrow \calh$ where $E_\psi(\calo) = \calo\ket{\psi}$ is continuous for all $\ket{\psi} \in \calh$. A sequence of bounded operators $\{\calo_n\}$ converges strongly if $\lim_{n \rightarrow \infty} \calo_n \ket{\psi}$ converges for all $\ket{\psi} \in \calh$.
\bigskip
\begin{defn}  \label{def:WOT}
The \emph{weak operator topology} is the smallest topology that contains the following basic neighborhoods
\begin{equation*}
\mathcal{N}(\calo,\ket{\xi_1},\ldots,\ket{\xi_n},\ket{\eta_1},\ldots,\ket{\eta_n},\epsilon) = \{\calp \in \calb(\calh): |\braket{\eta_i|(\calp - \calo)|\xi_i}| < \epsilon \quad \forall i \in \{1,2,\cdots,n\} \} .
\label{def:weakoperatortopology}
\end{equation*}
\end{defn}
The weak operator topology is then the weakest topology with a map $E_{\psi,\ell} : \calb(\calh) \rightarrow \mathbb{C}$ given by $E_{\psi,\ell}(\calo) = \ell(\calo \ket{\psi})$ that is continuous for all $\ket{\psi} \in \calh$ where $\ell \in \calh^*$ and $\calh^*$ is the dual of $\calh$. A sequence of bounded operators $\{\calo_n\}$ converges weakly if $\lim_{n \rightarrow \infty} \braket{\chi|\calo_n|\psi}$ converges for all $\ket{\chi},\ket{\psi} \in \calh$.

Moreover, we cite the following theorem of Banach and Alaoglu, which will be useful in the proof of Theorem \ref{thm:maintheoremCstar}.
\bigskip
\begin{thm}[Banach-Alaoglu \cite{Banach,Alaoglu}]
The unit ball of $\mathcal{B}(\mathcal{H})$ is compact for the weak operator topology.
\end{thm}
We will also use the following theorem, which will be helpful when we discuss separability issues:
\bigskip
\begin{thm}[Takesaki \cite{Takesaki}]
If $\mathcal{H}$ is separable (i.e. $\mathcal{H}$ has a dense countable part), bounded parts of $\mathcal{B}(\mathcal{H})$ have a countable basis of open sets for the strong and weak operator topologies.
\end{thm}
These topologies have some relations: norm convergence implies strong convergence, which in turn implies weak convergence. It is important to note that the weak and strong topologies are very different from the norm topology. In particular, they are not metrizable: no distance can be associated to them. As a result, their definitions involve explicitly the Hilbert space, unlike the one of the norm topology. As we will see, this fact explains the structural differences between the world of $C^*$-algebras and the one of von Neumann algebras.

\subsection{$C^*$-algebras}
As we want our framework to be well-suited to operator pushing, we need to introduce one of the two main objects of the theory of operator algebras: $C^*$-algebras. We follow notations from \cite{Takesaki} and \cite{BratteliRobinson} closely. A $C^*$-algebra can be defined as following.
\bigskip
\begin{defn}
A {\em $C^*$-algebra} is a complex algebra $\mathcal{A}$ equipped with a norm $\|.\|$ with respect to which it is complete, such that:
$$\forall A,B\in\mathcal{A},\;\|AB\|\leq\|A\|\|B\|,$$ 
and an involution $^*$ such that
\begin{align*}
\forall A,B\in\mathcal{A},\forall\lambda\in\mathbb{C},&\; (A+\lambda B)^*=A^*+\overline{\lambda}B^*,\\
\forall A\in\mathcal{A},&\; \|A^*A\|=\|A\|^2.
\end{align*}
\end{defn}
Note that in the definition of a $C^*$-algebra, no explicit mention is made of an action on a Hilbert space. The only structure needed is the norm, the algebra structure, and the involution. This makes $C^*$-algebras \textit{state-independent} objects. As such, they will be very well-adapted to study operator pushing in infinite tensor networks.

We also give a very important density theorem (this time for a $C^*$-algebra acting on a Hilbert space):
\bigskip
\begin{thm}[Kaplansky \cite{Kaplansky}]
Let $\mathcal{A}$ be a $C^*$-algebra acting on a Hilbert space, and $M$ be its strong operator closure. Then, the unit ball of $\mathcal{A}$ is dense in the unit ball of $M$ for the strong operator topology.
\end{thm}

%

\subsection{Von Neumann algebras} \label{sec:vNalg}
Another important object to consider in our setting will be von Neumann algebras. As we shall see, it will arise naturally once our Hilbert space is constructed, and we will see it as a state-dependent object.
\bigskip
\begin{defn}
	A \emph{$\star$-algebra} is an algebra of operators that is closed under hermitian conjugation.
\end{defn}
\bigskip
\begin{thm}[\cite{Jones-vNalg}, page 12]
	Let $M$ be a $\star$-subalgebra of $\mathcal{B}(\mathcal{H})$ that contains the identity operator. Let $M^{\prime\prime}$ be the double commutant of M, i.e. the algebras of operators that commute with all operators that commute with all operators of $M$. Then $M^{\prime \prime} = \overline{M}$, where the closure is taken for the strong operator topology.
\end{thm}
The following theorem is at the root of von Neumann algebra theory:
\bigskip
\begin{thm}[\cite{Jones-vNalg}, page 12]
\label{thm:doublecommutant}
	If $M$ is a $\star$-subalgebra of $\mathcal{B}(\mathcal{H})$ that contains the identity operator, then the following statements are equivalent:
	\begin{itemize}
		\item $M = M^{\prime \prime}$,
		\item $M$ is closed in the strong operator topology,
		\item $M$ is closed in the weak operator topology.
	\end{itemize}
\end{thm}
\bigskip
\begin{defn}
	A \emph{von Neumann algebra} is an algebra that satisfies the statements in Theorem \ref{thm:doublecommutant}.
\end{defn}

Note that given a $\star$-subalgebra of $\mathcal{B}(\mathcal{H})$ containing the identity, like, in particular, a $C^*$-algebra acting on a Hilbert space, we can generate a von Neumann algebra by taking either the double commutant or the closure in the strong or weak topology. This is how we shall proceed in the next section in order to get von Neumann algebras out of GNS representations of $C^*$-algebras. 

Unlike in the case of a $C^*$-algebra, these three characterizations (bicommutant, and strong and weak closures), are heavily dependent on the underlying Hilbert space structure, as strong and weak topologies are not metrizable and hence cannot be blind to the Hilbert space. As we shall see later, the construction of our Hilbert spaces will heavily depend on a choice of state. It follows that the structure of the von Neumann algebras on this Hilbert space will also heavily depend on this choice. It is in that sense that we will refer to von Neumann algebras as state-dependent objects.

%
%

\section{Tools from algebraic quantum field theory}
\label{sec:AQFT}
In a holographic quantum error correcting code, a code subspace is the space of states of a quantum field theory in curved spacetime on a fixed background geometry. As such, in the semiclassical regime, we expect bulk physics to be described by the tools developed since the 1960s in algebraic quantum field theory \cite{Haag}. We introduce some of these notions in a physics-friendly way relevant to our concerns (Theorem \ref{thm:maintheoremCstar} and its proof), with a particular emphasis on state dependence. We write in Theorem \ref{thm:modflow} a modified version of the Tomita--Takesaki theorem in \cite{Tomita-Takesaki} which incorporates thermality.

\subsection{$C^*$-algebras, von Neumann algebras, and state dependence}
In algebraic quantum field theory, the basic object of interest is a net of algebras of observables.\footnote{See Section \ref{sec:separableH} for the definition of net and details.} More precisely, to each open region of spacetime, one can associate an algebra of bounded operators, which is seen as the algebra of observables of the theory. One then requires that algebras of regions included in each other are included in each other, which means that they will inherit a net structure. The complexity of the theory is encoded in the structure of these inclusions, rather than the algebras themselves -- indeed, at least in the von Neumann context, local algebras all are hyperfinite factors of type $III_1$.

Usually, in algebraic quantum field theory, the local algebras are required to be $C^*$-algebras. One should remember that a $C^*$-algebra can be defined and constructed without ever resorting to a Hilbert space. In fact, this apparently innocent fact will turn out to be helpful in the case of the HaPPY code in Section \ref{sec:HaPPY}. 

With this in mind, we can naturally regard $C^*$-algebras as algebras of state-independent observables. This brings advantages as $C^*$-algebras are easier to construct, but it is not enough for considering holographic theories to omit entirely the aspect of von Neumamm algebras. Importantly, von Neumann algebras are defined through either a completion under strong or weak topologies, or a bicommutant, which capture the properties of the underlying Hilbert space representations. It follows that a von Neumann algebra intrinsically depends on the Hilbert space it acts on: it is a state-dependent object. Ultimately, we expect entanglement wedge reconstruction to exhibit at least some state-dependent features, which is the key reason why we expect von Neumann algebras to be the right objects to look at in the full quantum gravity regime. Hence we will see our $C^*$-algebras as a first step towards the von Neumann algebras. 

\subsection{The GNS representation}
For quantum mechanical systems we naturally represent quantum states as density matrices. This is innately from the perspective quantum mechanics from finite-dimensional Hilbert spaces, where the trace is easy to be manipulated. For infinite-dimensional Hilbert spaces, it is more natural to present the state in a more abstract manner, as an expectation value functional on the algebra of observables. This is equivalent to the density matrix picture in finite-dimensional cases, as a density matrix $\rho$ is uniquely determined by the associated expectation value functional $A\mapsto \mathrm{Tr}(\rho A)$. More rigorously, one can define as the following.
\bigskip
\begin{defn}
Let $\mathcal{A}$ be a $C^*$-algebra. A state on $\mathcal{A}$ is a linear functional $\omega$ on $\mathcal{A}$ such that for $A\in\mathcal{A}$, $\omega(A^\dagger A)$ is a nonnegative real number, and $\omega(Id)=1$.
\end{defn}
This definition enables us to define a state on an algebra of observables without having a Hilbert space. In order to build a Hilbert space directly from the $C^*$-algebra, we can use the Gelfand-Naimark-Segal (GNS) representation, following \cite{Pillet}. Given a state $\omega$, an intuitive idea to define the inner product is to take
\begin{align}\label{inner}
\braket{B,A}:=\omega(B^\dagger A)
\end{align}
for $A$ and $B$ in $\mathcal{A}$. However, this may not be possible as there can exist observables for which $\omega(A^\dagger A)=0$. Instead, one needs to consider the quotient $\mathcal{A}/I$ where
\begin{align}
I:=\{A\in\mathcal{A},\;\omega(A^\dagger A)=0\}.
\end{align}
This gets rid of the degeneracy and endows it with the structure of a Hilbert space onto which $\mathcal{A}$ acts by left multiplication. Utilizing these observations, the GNS representation can be presented more precisely as the following.
\bigskip
\begin{thm}[\cite{Pillet}]
Let $\mathcal{A}$ be a $C^*$-algebra and $\omega$ a state on $\mathcal{A}$. The space $\mathcal{H}_\omega:=\mathcal{A}/I$ is a Hilbert space for the inner product \eqref{inner}, on which $\mathcal{A}$ acts by left multiplication. This representation is called {\em the GNS representation} of $\mathcal{A}$. 
\end{thm} 
A key feature of the GNS representation $\pi_\omega$ of a state $\omega$ is that it purifies $\omega$ into the vector state $\ket{[Id]}=\ket{\Omega}$. Indeed, for all $A\in\mathcal{A}$, the vector state is given by
\begin{align}
\omega(A)=\bra{[Id]}\pi_\omega(A)\ket{[Id]}.
\end{align}

We shall see that this purification is nothing more than a generalization of the thermofield double construction, which is well-known by quantum information theorists. We will discuss this correspondence in more detail in section \ref{sec:TFD}.

The vector state $\ket{\Omega}$ has a nice property with respect to the $C^*$-algebra $\mathcal{A}$: it is cyclic. It means that the action of $\mathcal{A}$ on the state spans a norm-dense subset of the Hilbert space. Note that it is straightforward that this cyclic property will extend to the bicommutant of the $C^*$-algebra on the GNS Hilbert space, which is the von Neumann algebra it generates. 

\subsection{KMS states}
We now restrict our attention to a particular set of states, called KMS states. KMS states characterize thermal equilibrium with respect to a time evolution. In order to give some intuition, let us start with an $n$-dimensional system, the algebra $\mathcal{A}_n$ of $n$-dimensional matrices, and a Hamiltonian $H$ which generates the time evolution:
\begin{align}
\sigma_t(A):=e^{iHt}Ae^{-iHt}.
\end{align}
Then the Gibbs state 
\begin{align}
\rho:=\frac{e^{-\beta H}}{\mathrm{Tr}(e^{-\beta H})}
\end{align}
is the unique state representing thermal equilibrium at inverse temperature $\beta$. However, in infinite dimensions, traces are not always well-defined, which is why the Gibbs condition is no longer available.

One can prove that in finite dimensions, the Gibbs condition is equivalent to the more abstract KMS condition, which holds up to the infinite-dimensional case:
\bigskip
\begin{defn}
Let $\mathcal{A}$ be a $C^*$-algebra, $\omega$ a state on $\mathcal{A}$, and $\sigma_t$ be a one-parameter group of automorphisms of $\mathcal{A}$. For $\beta>0$, $\omega$ is a $\mathrm{KMS}_{\beta}$ state if for $A$ and $B$ in $\mathcal{A}$, there exists a function $F_{AB}$, analytic on the strip $\{0<\mathrm{Im}z<\beta\}$ and continuous on its closure, such that
\begin{align*}
F_{AB}(t)=\omega(A\sigma_t(B))\quad\text{and}\quad F_{AB}(t+i\beta)=\omega(\sigma_t(A)B).
\end{align*}
\end{defn}
For finite-dimensional settings, KMS states are naturally expected to exist. However, it is important to note that KMS states do not necessarily exist for all temperatures in the infinite-dimensional case, and are not necessarily unique when they exist. A change in the structure of the space of KMS states, which may or may not be unique, when lowering the temperature, is called spontaneous symmetry breaking, or a phase transition.

In this paper, KMS states and their GNS representations exhibit various interesting properties for our purposes in understanding holographic constructions. First, the GNS representation allows one to represent a KMS (thermal) state as a vector state on a Hilbert space, which deeply resonates with the thermofield double (TFD) constructions used in AdS/CFT \cite{Hartman:2013qma}. We discuss this in more detail in Section \ref{sec:TFD}. Second, the vector representative $\ket{\Psi}$ of a KMS state in its GNS representation is a separating vector for the $C^*$-algebra $\mathcal{A}$: if $A\in\mathcal{A}$ satistfies $A\ket{\Psi}=0$, then $A=0$.\footnote{Recall that cyclic and separating vectors play an important role for von Neumann algebras and that they were one of the most important ingredients in \cite{Kang:2018xqy} for obtaining an exact correspondence between entanglement wedge reconstruction and relative entropy conservation.} In fact, we get something even more broad: this property holds for the whole von Neumann algebra closure $\mathcal{A}''$. This is a very strong property which will turn out to be a key point in the proof of Theorem \ref{thm:maintheoremCstar}. In particular, it will allow us to derive the bulk-boundary entropy relations.

\subsection{Tomita--Takesaki theory} \label{sec:tomita}

When mention is made of thermality and the KMS condition in operator algebras, the crucial underlying tool is Tomita--Takesaki theory, which we now introduce following \cite{Araki,Witten:2018zxz,Jones-vNalg}.

\bigskip
\begin{defn}
	\label{def:cyc}
	A vector $\ket{\Psi} \in \calh$ is said to be {\em cyclic} with respect to a von Neumann algebra $M$ when the set of vectors $\calo\ket{\Psi}$ for $\calo \in M$ is dense in $\calh$. 
\end{defn}
\bigskip
\begin{defn}
	\label{def:sep}
	A vector $\ket{\Psi} \in \calh$ is {\em separating} with respect to a von Neumann algebra $M$ when zero is the only operator in $M$ that annihilates $\ket{\Psi}$. That is, $\calo\ket{\Psi} = 0 \implies \calo = 0$ for $\calo \in M$. 
\end{defn}

Given a von Neumann algebra $M \subset \calb(\calh)$ and a vector $\ket{\Psi} \in \calh$, we may define a map $e_\Psi:M \rightarrow \calh : \calo \mapsto \calo \ket{\Psi}$. $\calh$ is the closure of the image of $e_\Psi$ iff $\ket{\Psi}$ is cyclic with respect to $M$. Also, $\ker e_\Psi = \{0\}$\footnote{In other words, $e_\Psi$ is an injective map.} iff $\ket{\Psi}$ is separating with respect to $M$.

\bigskip
\begin{defn}
Let $\ket{\Psi},\ket{\Phi} \in \calh$ and $M$ be a von Neumann algebra. 
The {\em relative Tomita operator} is the operator $S_{\Psi | \Phi}$ that acts as
\begin{equation*}
S_{\Psi|\Phi} \ket{x} := \ket{y}
\end{equation*}
for any sequence $\{\calo_n\} \in M$ such that the limits $\ket{x} = \lim_{n \rightarrow \infty} \calo_n \ket{\Psi}$ and $\ket{y} = \lim_{n \rightarrow \infty} \calo_n^\dagger \ket{\Phi}$ both exist.
\end{defn}

The relative Tomita operator $S_{\Psi | \Phi}$ is well-defined on a dense subset of the Hilbert space if and only if $\ket{\Psi}$ is cyclic and separating with respect to $M$.\footnote{$S_{\Psi | \Phi}$ is well-defined if and only if $\lim_{n \rightarrow \infty} \calo_n \ket{\Psi} = 0 \implies \lim_{n \rightarrow \infty} \calo_n^\dagger \ket{\Psi} = 0$. See footnote 14 of \cite{Witten:2018zxz} for a proof of why this is true. $S_{\Psi | \Phi}$ is densely defined because $\ket{\Psi}$ is cyclic with respect to $M$.} Note that $S_{\Psi|\Phi}$ is a closed operator.

\bigskip
\begin{thm}[\cite{Jones-vNalg}, page 94]
	Let $\ket{\Psi},\ket{\Phi} \in \calh$ both be cyclic and separating with respect to a von Neumann algebra $M$. Let $S_{\Psi|\Phi}$ and $S^\prime_{\Psi|\Phi}$ be the relative Tomita operators defined with respect to $M$ and its commutant $M^\prime$ respectively. Then
	\begin{equation}
	S_{\Psi|\Phi}^\dagger = S^\prime_{\Psi|\Phi}, \ S_{\Psi|\Phi}^{\prime \, \dagger} = S_{\Psi|\Phi}.
	\end{equation}
\end{thm}
\bigskip
\begin{defn}
Let $S_{\Psi|\Phi}$ be a relative Tomita operator and $\ket{\Psi} \in \calh$ be cyclic and separating with respect to a von Neumann algebra $M$. The {\em relative modular operator} is $$\Delta_{\Psi|\Phi} := S_{\Psi|\Phi}^\dagger S_{\Psi|\Phi}.$$
\end{defn}

If $\ket{\Phi}$ is replaced with $\calo^\prime \ket{\Phi}$, where $\calo^\prime \in M^\prime$ is unitary, then the relative modular operator remains unchanged \cite{Witten:2018zxz}:
\begin{equation} \Delta_{\Psi|\Phi} = \Delta_{\Psi|\calo^\prime\Phi}. \end{equation}

\bigskip
\begin{defn}
	Let $M$ be a von Neumann algebra on $\calh$ and $\ket{\Psi}$ be a cyclic and separating vector for $M$. The \emph{Tomita operator} $S_\Psi$ is
	$$ S_\Psi := S_{\Psi | \Psi},$$ where $S_{\Psi | \Psi}$ is the relative modular operator defined with respect to $M$. The \emph{modular operator} $\Delta_\Psi = S_\Psi^\dagger S_\Psi$ and the \emph{antiunitary operator} $J_\Psi$ are the operators that appear in the polar decomposition of $S_\Psi$ such that
	$$S_\Psi = J_\Psi \Delta_{\Psi}^{1/2}.$$
\end{defn}
\bigskip
\begin{thm}[modified Tomita--Takesaki \cite{Tomita-Takesaki}]
	\label{thm:modflow}
	Let $M$ be a von Neumann algebra on $\calh$ and let $\ket{\Psi}$ be a cyclic and separating vector for $M$, let $\beta\in\mathbb{R}$. Then \begin{itemize}
		\item $J_\Psi M J_\Psi = M^\prime.$
		\item $\Delta_\Psi^{-\frac{it}{\beta}} M \Delta_\Psi^{\frac{it}{\beta}} = M \quad \forall t \in \mathbb{R}$.
	\end{itemize} 
	Moreover, $A\longmapsto\Delta_\Psi^{-\frac{it}{\beta}} M \Delta_\Psi^{\frac{it}{\beta}}$ defines the only one-parameter group of automorphisms of $M$ with respect to which $\Psi$ is a KMS$_\beta$ state.
\end{thm}

One of the findings of \cite{Jafferis:2015del} is that bulk modular flow is dual to boundary modular flow. We will also show this in our setup, as well as the fact that modular evolution at a given temperature will be directly traced back to Hamiltonian evolution.

\subsection{Araki's relative entropy} \label{sec:relent}

Previous works on entanglement entropy and AdS/CFT \cite{Jafferis:2015del,Casini2008,DongHarlowWall, casinitestetorroba2016} have used $S(\rho,\sigma) = \text{Tr }(\rho \log \rho - \rho \log \sigma)$ as the definition of relative entropy. However, there exists a more powerful definition, due to Araki, which involves Tomita--Takesaki theory, and can be extended to general von Neumann algebras. We use this definition in the rest of the paper. At the end, we briefly discuss an even broader definition of relative entropies, which could directly be used for states on $C^*$-algebras.

\bigskip
\begin{defn}[\cite{Araki}] \label{def:relent}
Let $\ket{\Psi},\ket{\Phi} \in \calh$ and $\ket{\Psi}$ be cyclic and separating with respect to a von Neumann algebra $M$. Let $\Delta_{\Psi|\Phi}$ be the relative modular operator. The {\em relative entropy} with respect to $M$ of $\ket{\Psi}$ is
	\begin{equation*} 
	\cals_{\Psi|\Phi}(M) = - \braket{\Psi|\log \Delta_{\Psi | \Phi}|\Psi}. 
	\end{equation*} 
\end{defn}

The relative entropy $\cals_{\Psi|\Phi}(M)$ is nonnegative and it vanishes precisely when $\ket{\Phi} = \calo^\prime \ket{\Psi}$ for a unitary $\calo^\prime \in M^\prime$.
\section{Infinite-dimensional entanglement wedge reconstruction} \label{sec:EWreconstruction}

Utilizing notions of von Neumann algebras and Araki's formalism of relative entropy, in this section we introduce the current knowledge about entanglement wedge reconstruction in AdS/CFT and discuss new insights and implications from our Theorem \ref{thm:maintheoremCstar}.

In the semiclassical limit where the bulk theory can be approximated as a quantum field theory on a fixed curved spacetime background, the bulk-to-boundary map can be described as a quantum error correcting code \cite{Almheiri:2014lwa}. In this context, it has been shown for finite-dimensional Hilbert spaces that exact bulk reconstruction is equivalent to the Ryu--Takayanagi formula, the conservation of relative entropies, and the conservation of modular flow between the bulk and the boundary. In particular, the equivalence has been proven rigorously by Harlow \cite{Harlow:2016vwg} for any finite-dimensional system. In infinite dimensions, we expect the Ryu--Takayanagi formula to have to be regulated in some way, but at least it was proven in \cite{Kang:2018xqy} that the equivalence between bulk reconstruction and the conservation of relative entropies still holds under some additional assumptions, which we recite below.

\bigskip
\begin{thm}[Kang-Kolchmeyer \cite{Kang:2018xqy}]
\label{thm:kangkolch}
Let $u : \calh_{code}\rightarrow \calh_{phys}$ be an isometry between two Hilbert spaces. Let $M_{code}$ and $M_{phys}$ be von Neumann algebras on $\calh_{code}$ and $\calh_{phys}$ respectively. Let $M^\prime_{code}$ and $M^\prime_{phys}$ respectively be the commutants of $M_{code}$ and $M_{phys}$. Suppose that the set of cyclic and separating vectors with respect to $M_{code}$ is dense in $\calh_{code}$. Also suppose that if $\ket{\Psi} \in \calh_{code}$ is cyclic and separating with respect to $M_{code}$, then $u \ket{\Psi}$ is cyclic and separating with respect to $M_{phys}$. Then the following two statements are equivalent:
\begin{description}
\item[ 1. Bulk reconstruction]
\begin{align} \nonumber
\begin{split}
\forall \calo \in M_{code}\ \forall \calo^\prime \in M_{code}^\prime, \quad 
\exists\tilde{\calo} \in M_{phys}\ \exists \tilde{\calo}^\prime \in M_{phys}^\prime\quad \text{such that}\quad\\
\forall \ket{\Theta} \in \calh_{code} \quad 
\begin{cases}
u \calo \ket{\Theta} =  \tilde{\calo} u \ket{\Theta}, \quad
&u \calo^\prime \ket{\Theta} =  \tilde{\calo}^\prime u \ket{\Theta}, \\
u \calo^\dagger \ket{\Theta} =  \tilde{\calo}^\dagger u \ket{\Theta}, \quad
&u \calo^{\prime \dagger} \ket{\Theta} = \tilde{\calo}^{\prime\dagger} u\ket{\Theta}.
\end{cases}\quad
\end{split}
\end{align}

\item[ 2. Boundary relative entropy equals bulk relative entropy]
\begin{align}\nonumber
\begin{split}
\text{For any $\ket{\Psi}$, $\ket{\Phi} \in \calh_{code}$ with $\ket{\Psi}$ cyclic }&\text{ and separating with respect to $M_{code}$,}\quad\quad\quad\\
\cals_{\Psi|\Phi}(M_{code})=\cals_{u\Psi|u\Phi}(M_{phys}),\ & \text{and} \  \cals_{\Psi|\Phi}(M_{code}^\prime)= \cals_{u\Psi|u\Phi}(M_{phys}^\prime),\\
\text{where $\cals_{\Psi|\Phi}(M)$ is the relative entropy.}\quad&
\end{split}
\end{align}
\end{description}
\end{thm}

In \cite{Kang:2018xqy}, it was further proven that the conservation of relative entropies follows from bulk reconstruction under milder assumptions:

\bigskip
\begin{thm}[Kang-Kolchmeyer \cite{Kang:2018xqy}]
\label{thm:kangkolchloose}
Let $u : \calh_{code}\rightarrow \calh_{phys}$ be an isometry between two Hilbert spaces. Let $M_{code}$ and $M_{phys}$ be von Neumann algebras on $\calh_{code}$ and $\calh_{phys}$ respectively. Let $M^\prime_{code}$ and $M^\prime_{phys}$ respectively be the commutants of $M_{code}$ and $M_{phys}$. 

\noindent Suppose that

\begin{itemize}
\item There exists some state $\ket{\Omega} \in \calh_{code}$ such that $u\ket{\Omega} \in \calh_{phys}$ is cyclic and separating with respect to $M_{phys}$. 
\item $\forall \calo \in M_{code}\ \forall \calo^\prime \in M_{code}^\prime, \quad 
\exists\tilde{\calo} \in M_{phys}\ \exists \tilde{\calo}^\prime \in M_{phys}^\prime$ such that
	\begin{align} \nonumber
	\begin{split}
	\forall \ket{\Theta} \in \calh_{code} \quad 
	\begin{cases}
	u \calo \ket{\Theta} =  \tilde{\calo} u \ket{\Theta}, \quad
	&u \calo^\prime \ket{\Theta} =  \tilde{\calo}^\prime u \ket{\Theta}, \\
	u \calo^\dagger \ket{\Theta} =  \tilde{\calo}^\dagger u \ket{\Theta}, \quad
	&u \calo^{\prime \dagger} \ket{\Theta} = \tilde{\calo}^{\prime\dagger} u \ket{\Theta}.
	\end{cases}
	\end{split}
	\end{align}
\end{itemize}

\noindent Then, for any $\ket{\Psi}$, $\ket{\Phi} \in \calh_{code}$ with $\ket{\Psi}$ cyclic and separating with respect to $M_{code}$, 
\begin{itemize}
\item $u \ket{\Psi}$ is cyclic and separating with respect to $M_{phys}$ and $M_{phys}^\prime$,
\item $\cals_{\Psi|\Phi}(M_{code})= \cals_{u\Psi|u\Phi}(M_{phys}), \quad \cals_{\Psi|\Phi}(M_{code}^\prime)= \cals_{u\Psi|u\Phi}(M_{phys}^\prime),$
\end{itemize}
where $\cals_{\Psi|\Phi}(M)$ is the relative entropy.
\end{thm}

This result is a suitable generalization of exact relations between bulk reconstruction and relative entropy equivalence between bulk and boundary in the finite-dimensional case to any (finite or infinite) von Neumann algebras without strong assumptions, but it still has limitations: it heavily relies on the existence of the isometry map $u$ between Hilbert spaces. In holographic codes defined by operator pushing rather than state pushing, such as the HaPPY code, it is actually difficult to directly construct such a map $u$ between Hilbert spaces. In this paper, we will develop a machinery to construct Hilbert spaces and an isometry $u$ directly out of a thermal state and an operator pushing map, and show that this construction is enough to deduce both bulk reconstruction and the conservation of relative entropies, hence giving a more general realization of the statement of Theorem \ref{thm:kangkolchloose}. 

\section{Main theorem and proof} \label{sec:main}
In this section, we state our main theorem (see Theorem \ref{thm:maintheoremCstar}) and motivate its setup along with its proof.

\subsection{The main theorem and its physical implications}
Our main theorem, given by Theorem \ref{thm:maintheoremCstar}, starts from having the bulk with its operator algebra as a $C^*$-algebra, but not necessarily having a Hilbert space. The theorem then provides as a result a Hilbert space construction, its vector states on the bulk, and the relative entropy equivalence between the boundary and the bulk. In short, this relaxes the conditions required for the entanglement wedge reconstruction  from Theorems \ref{thm:kangkolch} and \ref{thm:kangkolchloose} from \cite{Kang:2018xqy} and emphasizes the role of state-dependence.

\bigskip
\begin{thm:maintheoremCstar}
Let $\mathcal{A}_{code}$ and $\mathcal{A}_{phys}$ be two $C^*$-algebras, and let $\iota:\mathcal{A}_{code}\longrightarrow \mathcal{A}_{phys}$ be an isometric $C^*$-homomorphism. Let $\sigma_t$ be a strongly continuous one-parameter group of isometries of $\mathcal{A}_{phys}$ such that $\sigma_t(\iota(\mathcal{A}_{code}))\subset\iota(\mathcal{A}_{code})$, and $\omega$ be a KMS state on $\mathcal{A}_{phys}$ with respect to $\sigma_t$ at inverse temperature $\beta$. Then there exist Hilbert space representations $(\pi_\omega^{phys}, \mathcal{H}_{phys})$ and $(\pi_\omega^{code}, \mathcal{H}_{code})$ of $\mathcal{A}_{phys}$ and $\mathcal{A}_{code}$ such that:
\begin{enumerate}
\item there exists a Hilbert space isometry $u:\mathcal{H}_{code}\longrightarrow\mathcal{H}_{phys}$ such that
$$\forall A\in\mathcal{A}_{code},\quad \pi_\omega^{phys}(\iota(A))u=u\pi_\omega^{code}(A).$$

\item there exists a vector $\ket{\Omega}_{code}\in\mathcal{H}_{code}$ and a vector $\ket{\Omega}_{phys}\in\mathcal{H}_{phys}$ such that 
\begin{align*}
\forall A\in\mathcal{A}_{phys},\quad\omega(A)&=\bra{\Omega_{phys}}\pi_\omega^{phys}(A)\ket{\Omega_{phys}},\\
\forall A\in\mathcal{A}_{code},\quad\omega(\iota(A))&=\bra{\Omega_{code}}\pi_\omega^{code}(A)\ket{\Omega_{code}}.
\end{align*}

\item if $M_{code}=\pi_{\omega}^{code}(\mathcal{A}_{code})''$ and $M_{phys}=\pi_{\omega}^{phys}(\mathcal{A}_{phys})''$, then $\ket{\Omega_{code}}$ is cyclic and separating with respect to $M_{code}$ and $\ket{\Omega_{phys}}$ is cyclic and separating with respect to $M_{phys}$. Moreover,\begin{align} \nonumber
\begin{split}
\forall \calo \in M_{code}\ \forall \calo^\prime \in M_{code}^\prime, \quad 
\exists\tilde{\calo} \in M_{phys}\ \exists \tilde{\calo}^\prime \in M_{phys}^\prime\quad \text{such that}\quad\\
\forall \ket{\Theta} \in \calh_{code} \quad 
\begin{cases}
u \calo \ket{\Theta} =  \tilde{\calo} u \ket{\Theta}, \quad
&u \calo^\prime \ket{\Theta} =  \tilde{\calo}^\prime u \ket{\Theta}, \\
u \calo^\dagger \ket{\Theta} =  \tilde{\calo}^\dagger u \ket{\Theta}, \quad
&u \calo^{\prime \dagger} \ket{\Theta} = \tilde{\calo}^{\prime\dagger} u\ket{\Theta}.
\end{cases}\quad
\end{split}
\end{align}

\item if $\ket{\Phi}$ and $\ket{\Psi}$ are two vectors in $\mathcal{H}_{code}$ with $\ket{\Psi}$ cyclic and separating with respect to $M_{code}$, then $u\ket{\Psi}$ is cyclic and separating with respect to $M_{phys}$ and the equality of the relative entropy holds:
$$\mathcal{S}_{\Psi|\Phi}(M_{code})=\mathcal{S}_{u\Psi|u\Phi}(M_{phys}).$$

\item if $\Delta_{\Omega_{phys}}$ is the modular operator of $\ket{\Omega}_{phys}$ with respect to $M_{phys}$ and $\Delta_{\Omega_{code}}$ is the modular operator of $\ket{\Omega}_{code}$ with respect to $M_{code}$, then 
\begin{align*}
\forall A\in\mathcal{A}_{phys},\quad\pi_{\omega}^{phys}(\sigma_t(A))&=\Delta_{\Omega_{phys}}^{-\frac{it}{\beta}}\pi_{\omega}^{phys}(A)\Delta_{\Omega_{phys}}^\frac{it}{\beta},\\
\forall A\in\mathcal{A}_{code},\quad\pi_{\omega}^{phys}(\sigma_t(\iota(A)))u&=u\Delta_{\Omega_{code}}^{-\frac{it}{\beta}}\pi_{\omega}^{code}(A)\Delta_{\Omega_{code}}^{\frac{it}{\beta}}.
\end{align*}
\end{enumerate}
\end{thm:maintheoremCstar}

This result may seem complicated and unintuitive but it has a straightforward implication in physics. The first point of Theorem \ref{thm:maintheoremCstar} means that the $C^*$-isometry $\iota$ is \textit{implementable} in the considered representations. It follows that there exists a unitary operator which implements it at the level of the Hilbert spaces. The second point of Theorem \ref{thm:maintheoremCstar} shows that the constructed representations transform our KMS state and its pullback into vector states. The third point of Theorem \ref{thm:maintheoremCstar} shows that within this framework, bulk reconstruction is verified while the fourth point shows the conservation of relative entropies. Finally, the last point of Theorem \ref{thm:maintheoremCstar} relates bulk and boundary modular flows, in the spirit of JLMS \cite{Jafferis:2015del}.

Unlike its previous analogues, our result does not assume any pre-existing map between Hilbert spaces, and allows us to work directly at the level of the operators. Its physical meaning is also enhanced, as it makes a natural use of the boundary dynamics to construct the Hilbert spaces and von Neumann algebras. The boundary Hilbert spaces of states should be thought of as excitations of a thermal bath represented by the KMS state. The bulk Hilbert space can also have a similar interpretation, as the last item shows that the bulk and boundary modular flows coincide, therefore allowing $\ket{\Omega}_{code}$, which is thermal with respect to its bulk modular flow, to be thought of as thermal with respect to the flow of the system as a whole.

Note that a crucial assumption is the existence of a KMS state, which is not shown to be guaranteed for general quantum dynamical systems. However, it seems natural in physics that the dynamics are constructed is such a way that thermal equilibrium is possible at any finite temperature. In the particular case of trapeze dynamics on an infinite-dimensional HaPPY code which we will study later, the existence of such a state will come about quite naturally.

A physically more interesting question is whether the boundary KMS state is unique. In a theory with broken symmetry, we expect to have more than one choice of KMS state, yielding inequivalent GNS representations. However, our construction arbitrarily picks one of these representations, making the reasoning state-dependent from the very beginning. As we will only restrict ourselves to the setting of exact entanglement wedge reconstruction, this point will not need to be mentioned. However, in the full quantum gravity regime, we expect different Hilbert spaces, maybe corresponding to different entanglement wedges, to come into play simultaneously in an ensemble superposition. Such a picture is likely to be crucial for the emergence of gravity and the resolution of the black hole information paradox. We will return to this fascinating problem at the end of this paper and, hopefully, in future work. 

\subsection{Construction of the Hilbert spaces and the Hilbert space isometry} \label{sec:constructH}
Before proving the last point of the theorem, we will construct the representations $\pi_{\omega}^{phys}$ and $\pi_{\omega}^{code}$, and the mapping $u$. The representation construction is extremely simple once the algebraic quantum field theory techniques have been introduced: one simply needs to take $\pi_{\omega}^{phys}$ and $\pi_{\omega}^{code}$ to be the GNS representations of $\mathcal{A}_{phys}$ and $\mathcal{A}_{code}$ with respect to the state $\omega$ for $\pi_{\omega}^{phys}$, and the pullback state $\iota^*\omega$ for $\pi_{\omega}^{code}$, defined for $A\in\mathcal{A}_{code}$ by 
\begin{align}
\iota^*\omega(A):=\omega(\iota(A)).
\end{align}
Now let us define an isometric mapping $u:\mathcal{H}_{code}\longrightarrow \mathcal{H}_{phys}$. We know that 
\begin{align*}
\{\pi_{\omega}^{code}(A)\ket{\Omega}_{code},\;A\in\mathcal{A}_{code}\}
\end{align*}
is norm dense in $\mathcal{H}_{code}$. Let us define $u$ on this subspace by 
\begin{align}
u(\pi_{\omega}^{code}(A)\ket{\Omega}_{code})=\pi_{\omega}^{phys}(\iota(A))\ket{\Omega}_{phys}.
\end{align}
By the definition of the norm on the GNS representation, it is clear that this mapping is isometric. As it is a norm isometry, it extends to the whole Hilbert space as it maps Cauchy sequences to Cauchy sequences. Moreover, for $A\in\mathcal{A}_{code}$,
\begin{align}
\label{homomorphism}
\pi_{\omega}^{phys}(\iota(A))u=u\pi_{\omega}^{code}(A).
\end{align}
Indeed, let $A,O\in\mathcal{A}_{code}$. 
\begin{align*}
    \pi_\omega^{phys}(\iota(A))u\pi_\omega^{code}(O)\ket{\Omega}_{code}&=\pi_\omega^{phys}(\iota(A))\pi_\omega^{phys}(\iota(O))u\ket{\Omega}_{code}\\&=\pi_\omega^{phys}(\iota(AO))u\ket{\Omega}_{code}\\&=u\pi_\omega^{code}(A)\pi_\omega^{code}(O)\ket{\Omega}_{code}.
\end{align*}
Since the set $\{\pi_{\omega}^{code}(O)\ket{\Omega_{code}}\}$ being dense in $\mathcal{H}_{code}$, we have proved equation \eqref{homomorphism}.

\subsection{Proof of the theorem}
Section \ref{sec:constructH} proves the first point of Theorem \ref{thm:maintheoremCstar} and the second point simply follows from the definition of the GNS representation. We are left with the third and fourth points of Theorem \ref{thm:maintheoremCstar} to complete the proof.

In order to prove our claim about bulk reconstruction and relative entropies, we need to carefully check that the hypotheses of Theorem \ref{thm:kangkolch} (from \cite{Kang:2018xqy}) are satisfied. First, we need to prove that $\ket{\Omega}_{phys}=u\ket{\Omega}_{code}$ is cyclic and separating with respect to $M_{phys}$. This is a consequence of the fact that we are working in the GNS representation of a KMS state: as stated in section 2.4, GNS vector representatives of KMS states are always cyclic and separating both for the $C^*$-algebra and the induced von Neumann algebra. $\ket{\Omega}_{code}$ is also cyclic and separating for $M_{code}$. Indeed, it is easy to check, as $\mathrm{Im}(\iota)$ is stabilized by $\sigma_t$, that $\iota^*\omega$ is KMS for the time evolution defined on $\mathcal{A}_{code}$ 
\begin{equation}
\sigma_t^c(A):=\iota^{-1}(\sigma_t(\iota(A))).
\end{equation}
This time evolution is still an isometry, as $\mathrm{Im}(\iota)$ is stabilized by $\sigma_t$. Therefore, $\ket{\Omega_{code}}$ is cyclic and separating for the von Neumann algebra $\mathcal{M}_{code}$, by the same argument from section 2.4.

Then, for any $\mathcal{O}\in \mathcal{A}_{code}$, $\pi_{\omega}^{phys}(\iota(\calo))$ satisfies the requirements of $\tilde{\calo}.$ Then, if $\mathcal{O}\in M_{code}$, let $\mathcal{O}_\alpha\in A_{code}$ such that $(\pi_\omega^{code}(\mathcal{O}_\alpha))$ is of uniformly bounded norm and converges strongly towards $\mathcal{O}$. The existence of such a net\footnote{By net, we here mean that $\mathcal{O}_\alpha$ is indexed with $\alpha$ in a directed set. We return to the fact we have to use of nets rather than sequences in the next subsection.} is guaranteed by Kaplansky's density theorem. Then, $\pi_{\omega}^{phys}(\iota(\calo_\alpha))$ is of uniformly bounded norm, therefore, by the Banach-Alaoglu theorem, one can restrict one's attention to a subnet without loss of generality and suppose that $\pi_{\omega}^{phys}(\iota(\calo_\alpha))$ is weakly convergent towards $\tilde{\mathcal{O}}\in M_{phys}$. Then, for all  $\ket{\Theta}\in\mathcal{H}_{code}$ and $\ket{\chi}\in\mathcal{H}_{phys}$, 
\begin{align}
\bra{\chi}u\pi_\omega^{code}(\mathcal{O}_\alpha)\ket{\Theta}=\bra{\chi}\pi_{\omega}^{phys}(\iota(\calo_\alpha))u\ket{\Theta}.
\end{align}
By going to the limit, we obtain 
\begin{align}
u\pi_\omega^{code}(O)=\tilde{O}u.
\end{align}
Moreover, Hermitian conjugation being weakly continuous, $(\pi_{\omega}^{phys}(\iota(\calo_\alpha^\dagger)))=\pi_{\omega}^{phys}(\iota(\calo_\alpha))^\dagger$ (as $\iota$ is a $C^*$-homomorphism) converges weakly towards $\tilde{\mathcal{O}^\dagger}\in M_{phys}$. Then, 
\begin{align}
u\pi_\omega^{code}(O^\dagger)=\tilde{O}^\dagger u.
\end{align}
This proves that bulk reconstruction is possible for $M_{code}$ and $M_{phys}$.

Let us now turn to the commutants. For this, we will need the following lemma.
\bigskip
\begin{lem}
Let $J_{code}:=J_{\Omega_{code}}$ and $J_{phys}:=J_{\Omega_{phys}}$. Then,
\begin{align}
J_{phys}u=uJ_{code}.
\end{align}
\end{lem}
\begin{proof}
It is a direct consequence of the findings of \cite{Kang:2018xqy}. Note that 
\begin{align}
J_{phys}=S_{phys}\Delta_{phys}^{-\frac{1}{2}}\quad\text{and}\quad J_{code}=S_{code}\Delta_{code}^{-\frac{1}{2}}.
\end{align}
By \cite[p.17]{Kang:2018xqy}, we know that 
\begin{align}
u\Delta_{code}=\Delta_{phys}u\quad\text{and}\quad uS_{code}=S_{phys}u.
\end{align} 
We also have that 
\begin{align}
u\Delta_{code}^{-\frac{1}{2}}=\Delta_{phys}^{-\frac{1}{2}}u.
\end{align}
Indeed, $\Delta_{phys}$ and $\Delta_{code}$ have a functional calculus, and in \cite[p.18]{Kang:2018xqy} it is proved that $u$ commutes with all the projection valued measures. Then, the result is obvious.
\end{proof}
Now, let $\calo^\prime\in M_{code}^\prime$. Thanks to the Tomita--Takesaki theorem, there exists $\calo\in M_{code}$ such that \begin{align}
\calo^\prime=J_{code}\calo J_{code}.
\end{align} 
Then, let us define the operator
\begin{align}
\tilde{\calo}^\prime:=J_{phys}\tilde{\calo} J_{phys},\quad \tilde{\calo}^\prime\in M_{phys}^\prime.
\end{align}
Then these operators with the isometry satisfy the following relation:
\begin{align}
\tilde{\calo^\prime}u=J_{phys}\tilde{\calo}J_{phys}u=J_{phys}\tilde{\calo} u J_{code}=J_{phys}u \calo J_{code}=u J_{code}\calo J_{code}=u\calo^\prime.
\end{align}
Note that this construction also maps $\mathcal{O}^\dagger$ to $\tilde{\mathcal{O}}^\dagger$.

Therefore, the bulk reconstruction hypothesis of Theorem \ref{thm:kangkolchloose} from \cite{Kang:2018xqy} is satisfied. Hence, the third point is proven. 

Furthermore, the cyclic and separating conditions assure that we can conclude that relative entropies are conserved, which enables us to complete the proof of the fourth point.

Let us now prove the first equality of the fifth and last point. For all $t$, 
\begin{align*}
\sigma_t(\mathcal{A}_{phys})=\mathcal{A}_{phys}.
\end{align*}
By uniqueness of the GNS representation up to unitaries, and since $\omega$ is invariant by $\sigma_t$, it follows that for all $t$, there exists a unitary operator $U_t$ on $\mathcal{H}_{phys}$ such that 
\begin{align*}
\pi_\omega^{phys}(\sigma_t(A))\ket{\Omega_{phys}}=U_t^\dagger\pi_\omega(A)\ket{\Omega_{phys}}.
\end{align*}
Then the unitary $U_t$ therefore extends $\sigma_t$ to the whole von Neumann algebra $M_{phys}$ through $A\longmapsto U_t^\dagger A U_t$. We now need to show that $\ket{\Omega}_{phys}$ defines a KMS state with respect to the whole von Neumann algebra for this new time evolution. For this, we will need to use Kaplansky's density theorem and a standard argument, closely following the proof of Proposition 12 of \cite{Pillet}:
Let $A, B\in M_{phys}$ such that $\|A\|\leq 1$ and $\|B\|\leq 1$. By Kaplansky's density theorem, there exists nets $(A_\alpha)$, $(B_\alpha)$ of operators in $\pi_\omega(\mathcal{A}_{phys})$ which strongly converge to $A$ and $B$ respectively. Let \begin{align}
d_\alpha:=\mathrm{max}(\|(A_\alpha-A)\ket{\Omega_{phys}}\|,\|(A^\dagger_\alpha-A^\dagger)\ket{\Omega_{phys}}\|,\|(B_\alpha-B)\ket{\Omega_{phys}}\|,\|(B^\dagger_\alpha-B^\dagger)\ket{\Omega_{phys}}\|).
\end{align}
Then, by definition of the nets, $\underset{\alpha}{\lim}\;d_\alpha=0.$ As $\ket{\Omega_{phys}}$ is a KMS state on $\mathcal{A}_{phys}$, there exists a function $F_\alpha$ which is holomorphic on the strip $\{0<\mathrm{Im}z<\beta\}$ and bounded on its closure, and that satisfies
\begin{align}
\label{1}
F_\alpha(t)&=\bra{\Omega_{phys}}A_\alpha U_t^\dagger B_\alpha U_t\ket{\Omega_{phys}},\\
\label{2}
F_\alpha(t+i\beta)&=\bra{\Omega_{phys}}U_t^\dagger B_\alpha U_tA_\alpha\ket{\Omega_{phys}},
\end{align}
and by the Cauchy--Schwarz inequality,
\begin{align}
\|F_\alpha(t)-F_{\alpha^\prime}(t)\|&\leq 2(d_\alpha+d_\alpha^\prime),\\
\|F_\alpha(t+i\beta)-F_{\alpha^\prime}(t+i\beta)\|&\leq 2(d_\alpha+d_\alpha^\prime).
\end{align} 
Then, by Hadamard's three line theorem \cite{Hadamard},
\begin{equation}
\underset{0\leq\mathrm{Im}z\leq\beta}{\mathrm{sup}}\abs{F_\alpha(z)-F_{\alpha^\prime}(z)}\leq 2(d_\alpha+d_{\alpha^\prime}).
\end{equation}
This allows us to conclude that $(F_\alpha)$ is Cauchy for the uniform convergence, so it converges uniformly to some $F$ which is analytic on $\{0<\mathrm{Im}z<\beta\}$ and continuous on its closure. Moreover, by going to the limit in \eqref{1} and \eqref{2}, one finally gets
\begin{align}
F(t)&=\bra{\Omega_{phys}}A U_t^\dagger B U_t\ket{\Omega_{phys}},\\
F(t+i\beta)&=\bra{\Omega_{phys}}U_t^\dagger B U_tA\ket{\Omega_{phys}}.
\end{align} 
As for any two elements of $M_{phys}$, one can always divide by the maximum of their norm so that they both lie in the unit ball, we therefore conclude that this proof shows that $\ket{\Omega}_{phys}$ defines a KMS state on the whole von Neumann algebra $M_{phys}$.
Then, by the uniqueness part of the Tomita--Takesaki theorem, we conclude that the automorphism defined by $U_t$ coincides with the modular evolution at inverse temperature $\beta$, which concludes the proof of the first equality. The second equality of the fifth point follows from the first point together once again with the fact, proven in \cite{Kang:2018xqy}, that $\Delta_{phys}$ and $\Delta_{code}$ have a functional calculus\footnote{More precisely, it is proven in \cite[p.18]{Kang:2018xqy} that $uu^{\dagger}$ commutes with the physical projections.} and that $u$ commutes with all the projection valued measures.

\section{Special case: separable Hilbert spaces}
\label{sec:separableH}
For the proof of the Theorem \ref{thm:maintheoremCstar}, we used \textit{nets} instead of sequences when taking approximations of operators.\footnote{Nets are a mathematical notion generalizing the use of sequences.} Nets can rigorously be defined in the following way:
\bigskip
\begin{defn}
A directed set is an ordered set such that every pair of elements admits a common upper bound. A net of bounded operators is a family of bounded operators indexed by a directed set.
\end{defn}

The necessity for the use of nets comes from the fact that the strong and weak operator topologies may not have a countable basis of open sets on bounded parts of $\mathcal{B}(\mathcal{H})$. 

However, if we suppose our $C^*$-algebras to be \textit{separable} (i.e. there exists a countable dense subset for the norm topology), a result of \cite{Takesaki} can allow us to use sequences, as the Hilbert space representations will a fortiori be separable. The result of \cite{Kang:2018xqy} is proven with sequences, therefore it assumes separable bulk and boundary Hilbert spaces. However, a similar proof would still hold if nets are used instead of sequences.

For the cases of separable Hilbert spaces, Theorem \ref{thm:maintheoremCstar} can be proven using sequences. Whenever a net $(\mathcal{O}_\alpha)$ of approximating operators inside $\mathcal{A}_{code}$ or $\mathcal{A}_{phys}$ is used, replacing $\alpha$, which takes values in a directed set, with an integer index $n$, implies the same conclusions.

In the case of the HaPPY code, the $C^*$-algebras are separable in the bulk. It follows that the Hilbert spaces are always separable for the HaPPY code. This is expected, as the infinite-dimensional HaPPY code is a countably infinite pentagon tiling, which itself is a sequence of finite layers. Therefore, proofs with sequences will work in this case.

Nevertheless, we want to emphasize here that the use of nets instead of sequences might be of importance in some systems, namely those for which the Hilbert space is not separable. These systems do arise in the context of algebraic quantum field theory in presence of a continuous superselection rule \cite{Halvorson}. However, we expect each boundary superselection sector to have a separable Hilbert space, provided it does not violate the split property \cite{Haag}.

\section{Link with the thermofield double and the two-sided black hole} \label{sec:TFD}
In this section, we interpret our abstract construction in physical terms. In particular, we explicitly show that a GNS representation of a KMS state can be seen as a thermofield double construction, as well as how the commutants of the von Neumann algebras can be interpreted as the other side of a wormhole, in a similar spirit to the mirror operators of Papadodimas and Raju \cite{Papadodimas:2013wnh,Papadodimas:2013jku,Papadodimas:2012aq}.

\subsection{The thermofield double as a GNS representation}
We now introduce the link between our construction and the thermofield double in AdS/CFT. We first consider the thermofield double construction. In a finite-dimensional Hilbert space $\mathcal{H}$, a thermal density matrix is defined as
\begin{align}
\rho_{\beta}:=\frac{e^{-\beta H}}{\mathrm{Tr}{e^{-\beta H}}},
\end{align}
where $\rho_\beta$ denotes a mixed thermal state at temperature $\beta$. 
The thermofield double construction then doubles the size of the Hilbert space to construct a purification of the mixed $\rho_\beta$. In $\mathcal{H}\otimes\mathcal{H}$, which is the doubled Hilbert space, the vector state is constructed as
\begin{align}
\ket{TFD}_\beta:=\sum_ie^{-\beta\frac{E_i}{2}}\ket{e_i}\otimes\ket{e_i},
\end{align}
where the $\ket{e_i}$ denotes an eigenbasis of the Hamiltonian. This thermofield double vector state defines a purification of the thermal density matrix $\rho_\beta$. In spirit, what the thermofield double does in finite dimensions is that it transforms any mixed thermal state into a vector state. This is exactly what the GNS representation of a KMS state does.

To be more precise, let us consider the finite-dimensional context with an $n$-dimensional Hilbert space $\mathcal{H}_n$. For a Gibbs state, as long as the energy levels are nondegenerate, which we will suppose for simplicity, the thermal density matrix is invertible, which means that the ideal of the GNS representation is trivial. The Hilbert space is therefore directly constructed out of the algebra $\mathcal{M}_n(\mathbb{C})$, which is isomorphic to $\mathcal{H}_n\otimes\mathcal{H}_n$, and the density matrix $\rho_\beta$ is now represented as a vector state. As a conclusion, the GNS representation does exactly the same thing as the thermofield double in the nondegenerate finite-dimensional case.

The interpretation usually given to the thermofield double construction in AdS/CFT is that it constructs the second boundary of a double-sided AdS wormhole. It is therefore tempting to give the same meaning to the GNS representation. We shall explore these links in further detail in the next subsection.

Note that, throughout this discussion, we have carefully omitted to call vector states pure states. The reason is that purity becomes a trickier concept in infinite dimensions, and not all vector states are pure. Without going into any detail, we just stress the fact that in infinite dimensions, not all vector states are extremal points of the convex space of expectation value functionals on an operator algebra. Our construction therefore has no reason to give pure states in that sense.

\subsection{An explicit realization of the Papadodimas-Raju proposal}

In finite dimensions, we have seen that doing a GNS representation of a matrix algebra on a Hilbert space (i.e. a finite-dimensional factor) with respect to a nondegenerate thermal state amounts to doubling the number of degrees of freedom in the Hilbert space to realize the thermal state as a vector state. It is then usual to interpret it as the construction of the other side of a wormhole. 

However, in the infinite-dimensional case, things are not that simple as it is no longer generally possible to factorize the Hilbert space. Rather than seeing the two sides of the wormhole as encoded by two tensor factors of the Hilbert space, it is easier to consider them at the level of the algebras of observables: one side of the wormhole corresponds to the von Neumann algebras $M_{code}$ and $M_{phys}$, while the other one corresponds to their commutants $M^\prime_{code}$ and $M^\prime_{phys}$. As we saw in the proof of our main result, these commutants are assured to be isomorphic to the initial algebras, due to the Tomita--Takesaki theorem. Therefore, one can see these as encoding an identical copy of the physical system: here, the other side of the wormhole. 

This idea goes back to the work of Papadodimas and Raju, who argue that using modular conjugation in the exact same fashion as we did can make the black hole interior's reconstruction possible \cite{Papadodimas:2013wnh,Papadodimas:2013jku,Papadodimas:2012aq}. They emphasize that their construction is state-dependent, as they have to choose a reference state - which would here correspond to our initial choice of a KMS state. It is therefore interesting to see that the Papadodimas-Raju proposal appears in a very canonical way through the GNS representation.

Another issue that arises in the Papadodimas-Raju proposal is that the object they take the commutant of is not actually an algebra. We might have to face similar technical complications when we try to prove an approximate result on entanglement wedge reconstruction, instead of an exact one. In particular, we expect very excited states (i.e. operators of high energy acting on the cyclic and separating KMS state) to be hard to reconstruct with a good precision, due to the $\frac{1}{N}$ corrections. This could introduce a new level of state-dependence in our construction. We leave such investigations on approximate entanglement wedge reconstruction to future work.

\section{An example: infinite-dimensional HaPPY codes} \label{sec:HaPPY}
We now apply our result to study an infinite-dimensional analogue of the HaPPY code. We will use some of the results proven in our companion paper \cite{MonicaElliott}, where we conduct an extensive study of a model of dynamics. However, our setup will be a bit different. We start by briefly reviewing the model of \cite{MonicaElliott}, before introducing our new setup through reverse engineering. We then define our $C^*$-algebras and their dynamics, for which we prove the existence of KMS states at all temperatures. We conclude by constructing our Hilbert spaces and associated von Neumann algebras, and proving the conservation of relative entropies with our theorem. 

\subsection{The infinite-dimensional HaPPY code: a review}
Here, we explain how the HaPPY code levels are recursively constructed  directly at the {\it level} of the Hilbert spaces in \cite{MonicaElliott}. For a general review of the HaPPY code, see \cite{Pastawski:2015qua}.

We first define the notion of the {\it {level n}}. Given level 1 tile, level 2 tiles are defined to be all the (pentagon) tiles that share any vertices and/or edges with level 1 tile. Likewise, level 3 tiles are defined to be all tiles (except for tiles that are already in levels 1 or 2) that share any edges and/or vertices with level 2 tiles.

While the geodesic lines separate individual tiles in the Poincare disc into pentagons, for the purpose of the tensor network, we can start from the center as the first pentagon tile. Every tile in the pentagon tiling from the center towards the boundary corresponds to precisely one single level. Henceforth, level 1 consists simply of the center pentagon, as represented in blue in Figure \ref{fig:level2tensornetwork}.

\begin{figure}[H]
	\vspace{8mm}
	\centering
	\includegraphics[width=0.68\linewidth]{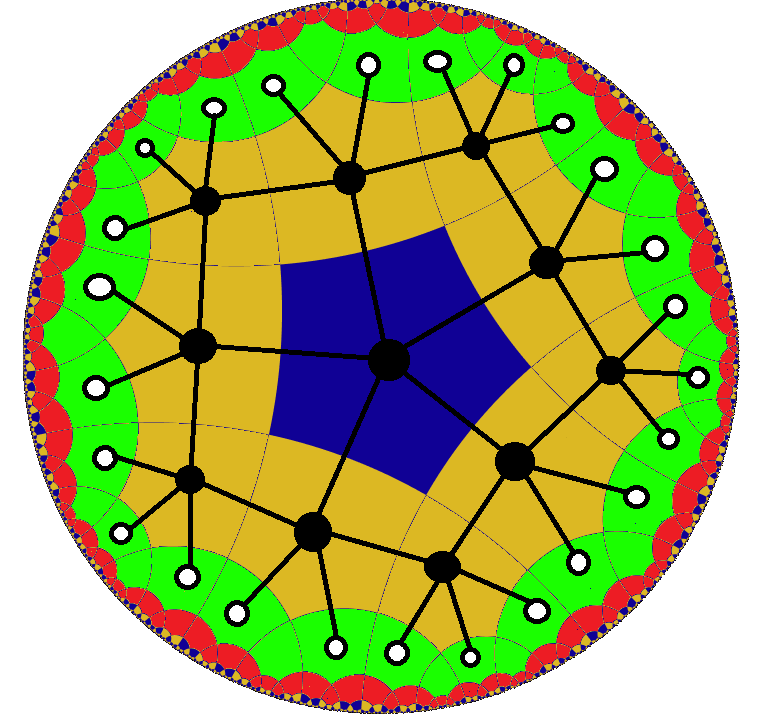}
	\vspace{5mm}
	\caption{This is the HaPPY code at level 2 with bulk and boundary nodes. The black qubits are the bulk nodes and the white dangling qubits are the boundary nodes. The blue region corresponds to level 1 bulk, which is represented with a single central node. The yellow regions are the new additional tiles for the bulk of level 2. (Level $n$ bulk always include level $n-1$ bulk. For example, the blue tile from level 1 is also a part of bulk for level $2$ and beyond.) The green regions are the addendums for the bulk of level 3. The red regions are the addendums of the bulk of level 4.}
	\label{fig:level2tensornetwork}
\end{figure}

Now that we know how levels are defined, we will describe how each level can be used to produce a state of finitely-many qubits. First, start at level 1. Put a bulk qubit in level 1 tile. Then, draw boundary legs through each edge of level 1 tile.

The {\it {level 1 tensor network}} is defined to be a map from one bulk qubit to five boundary qubits. We denote by $\widetilde{ N}_n$ the number of boundary qubits of level $n$ tensor network. In the case of level 1, $\widetilde{N} _1 = 5$. For example, we can put the bulk qubit in the state $\ket{\downarrow}$. Then, level 1 tensor network maps a given bulk state to a particular state of $\widetilde{N}_1 =5$ boundary qubits, which we call {\it {level 1 boundary state.}}

The {\it {level n tensor network}} is constructed as following. We first put a bulk qubit in every tile in level $n$. Next, we draw a connecting leg through any edge that borders at least one tile with a bulk qubit. Of course, each edge borders two tiles. If a leg connects two bulk qubits, then it defines a contraction of two tensor indices. However, if a leg connects a bulk qubit with a tile in level $n+1$, then it is a ``dangling leg'' and thus defines a boundary qubit. For example, the level 2 tensor network is represented in Figure \ref{fig:level2tensornetwork}.
 
Each choice of bulk state at a given level defines a boundary state at the same level.  For example, let us suppose that all bulk qubits are in the state $\ket{\downarrow}$. This defines a state of $\widetilde{N}_n$ boundary qubits at every level $n$.

Let $\calh_n$ for $n \in \mathbb{N}$ denote the Hilbert space of the bulk qubits at level $n$. For all $n \in \mathbb{N}$, $\calh_n$ is identified with (or, mapped into) the subspace of $\calh_{n+1}$ where all the qubits in level $n+1$ are in a reference state, which for now we may take to be the state $\ket{\downarrow}$. The relationship between the $\calh_n$ Hilbert spaces is expressed as
\begin{align*}
\calh_1 \rightarrow \calh_2 \rightarrow \calh_3 \rightarrow \cdots \quad ,
\end{align*}
where each arrow $\rightarrow$ denotes the map from a Hilbert space into a Hilbert space of the next level. We call this $\calh_n$ the {\it code pre-Hilbert space at level $n$}.

For each $\calh_n$, we can use the tensor network at level $n$ to map the bulk qubits at level $n$ into their corresponding white (boundary) qubits. The white qubits at level $n$ make up the {\it physical pre-Hilbert space at level $n$}, or $\widetilde{\calh}_n$. We represent the isometry produced by the tensor network with a down arrow, which can be expressed as
\begin{align*}
\begin{array}{cc}
\calh_n & \\ 
\downarrow & \\ 
\widetilde{\calh}_n &  .
\end{array}
\end{align*}
Putting these maps together, we establish the following:
\begin{align*}
\begin{array}{cccccccc}
\calh_1 & \rightarrow & \calh_2 & \rightarrow & \calh_3 & \rightarrow & \cdots & \\
\downarrow & & \downarrow & & \downarrow & & \downarrow & \\ 
\widetilde{\calh}_1 &  & \widetilde{\calh}_2 &  & \widetilde{\calh}_3 &  & \cdots & \quad .
\end{array}
\end{align*}

The missing piece is the boundary-to-boundary map, which is the isometric map which takes a state in $\widetilde{\calh}_n$ into a state in $\widetilde{\calh}_{n + 1}$. This can be constructed by putting tensors in all tiles in level $n+1$ and using the tensor network to map the boundary qubits in level $n$ to the boundary qubits in level $n+1$. The bulk indices of these tensors (which correspond to bulk qubits) must be taken care of in a systematic way.  For example, if they are all put in the reference state $\ket{\downarrow}$, the boundary map is constructed in such a way that the following diagram is commutative:
\begin{align*}
\begin{array}{cccccccc}
\calh_1 & \rightarrow & \calh_2 & \rightarrow & \calh_3 & \rightarrow & \cdots & \\
\downarrow & & \downarrow & & \downarrow & & \downarrow & \\ 
\widetilde{\calh}_1 & \rightarrow & \widetilde{\calh}_2 & \rightarrow & \widetilde{\calh}_3 &  & \cdots & \quad .
\end{array}
\end{align*}

For both the code and physical pre-Hilbert spaces, the right arrow (bulk-to-bulk or boundary-to-boundary map) means that a Hilbert space is isometrically mapped into, and identified with a subspace of, the Hilbert space of the next level. The problem of such an approach is that it will typically create very complicated maps between the boundary Hilbert spaces that will be impossible to be explicitly generated. In \cite{MonicaElliott}, we construct an explicit sequence of maps between boundaries using Bell pairs, in such a way that the entanglement in the bulk can be tracked.

The {\it code pre-Hilbert space} is the disjoint union of all the $\calh_n$ quotiented by the equivalence relation that relates two states if one is the image of the other by the inclusion map. Alternatively, the code pre-Hilbert space is the set of states where all but finitely many bulk qubits are in the reference state. The {\it physical pre-Hilbert space} is defined as the disjoint union of all the $\widetilde{\calh}_n$ quotiented by the equivalence relation that relates two boundary states if one is the image of the other by the bulk-to-boundary map, with the bulk in the reference state. Note that there is no way to define a physical reference state independently from a bulk reference state. Finally, the bulk and boundary Hilbert spaces are obtained by taking the norm completions of the pre-Hilbert spaces.\footnote{In mathematical terms, we have constructed direct limit Hilbert spaces for the bulk and the boundary, in a similar fashion as the semicontinuous limit described in \cite{Osborne:2017woa}.}

One of the advantages of this approach is that it reproduces the physics of the bulk all the way to the AdS radius.
This is relevant when considering deep bulk objects like black holes. On the other hand, it is not clear if the successive mappings of a bulk state will converge in some sense all the way up to the boundary. As the HaPPY code is a stronger approximation of the physics close to the boundary than deep in the bulk, since the dominant portion of the bulk nodes are adjacent to the boundary, the physical interpretation of this picture remains debatable.

Another problem with this picture is that it is constructed at the level of the Hilbert spaces. If we want to map Hamiltonian theories all the way up to the boundary, we will need to perform \textit{operator pushing} instead of \textit{state pushing}. In \cite{MonicaElliott}, we were successful in studying operator-pushing for a particular bulk Hamiltonian through the successive levels of the HaPPY code. We will revisit our results in this section, and formulate entanglement wedge reconstruction at the level of the $C^\ast$-algebras of bulk and boundary observables using Theorem \ref{thm:maintheoremCstar} in the slightly different setup of reverse engineering. The Hilbert space will then be constructed through a GNS representation, providing an alternative to the reference state approach.

\subsection{Reverse engineering for the HaPPY code}
The main drawback of the setting of  \cite{MonicaElliott} is that a bulk operator on a finite number of qubits at a finite distance from each other will be sent through operator pushing to increasingly distant qubits as the boundary grows. It will therefore be hard to keep track of explicit bulk-to-boundary mappings in that context.

Another way of defining an infinite-dimensional entanglement wedge in the HaPPY code, is to reason the other way around: instead of starting from the center the bulk and extending the boundary all the way out, one can consider an infinite string of qubits on the boundary as the Hilbert space of a given entanglement wedge, and reconstruct the bulk layer by layer, in a way that is consistent with the layer structure of the HaPPY code. We shall call such a process \textit{reverse engineering}.

More specifically, we shall restrict ourselves to an entanglement wedge of a trapeze shape in the bulk, as shown on Figure \ref{reverse}. Figure \ref{reverse} illustrates that at any finite cut of the HaPPY code, such a wedge can be reconstructed by starting from the boundary, superposing trapezes made of 2-clusters (in red), and constructing the code underneath layer by layer (in blue then in grey), by adding 2-clusters on the bottom of the legs, and 3-clusters in the middle of 2-clusters. Our method is to take this idea seriously and to expand this construction up to an infinite number of bulk layers.

In terms of physics, this just means that we take the viewpoint that the HaPPY code is more adapted to describe the physics near the boundary than deep in the bulk. A downside is that black holes and wormholes will not be possible to immediately see within that framework, as the construction will never reach the center of the bulk. 

\begin{figure}[H]
\centering
\begin{subfigure}{1\textwidth}
\begin{tikzpicture}
\node[draw,circle,thick,scale=1,black,label={[label distance=1mm]north:Center}] (C) at (0,2) {};
\draw (-4,0)--(C)--(2,3);
\draw (4,0)--(C)--(-2,3);
\draw (0,4)--(C);
\draw (-8,0)--(8,0);
\draw (-8,-2.5)--(8,-2.5);
\draw (-8,-4.5)--(8,-4.5);
\draw[draw,line width=1mm] (-4,0)--(4,0);
\draw[draw,line width=1mm] (-6,-2.5)--(-4,0);
\draw (-4,0)--(-2,-2.5);
\draw (2,-2.5)--(4,0);
\draw[draw,line width=1mm] (4,0)--(6,-2.5);
\draw[draw,line width=1mm] (-6,-2.5)--(-4,0);
\draw (-4,0)--(-2,-2.5);
\draw[draw,line width=1mm] (-6,-2.5)--(-6-0.9,-4.5);
\draw (-6,-2.5)--(-6+0.5,-4.5);
\draw (-2+0.9,-4.5)--(-2,-2.5);
\draw (-2-0.5,-4.5)--(-2,-2.5);
\draw (-4,-2.5)--(-4+0.9,-4.5);
\draw (-4,-2.5)--(-4,-4.5);
\draw (-4,-2.5)--(-4-0.9,-4.5);
\draw[draw,line width=1mm] (6,-2.5)--(6+0.9,-4.5);
\draw (6,-2.5)--(6-0.5,-4.5);
\draw (2+0.5,-4.5)--(2,-2.5);
\draw (2-0.9,-4.5)--(2,-2.5);
\draw (4,-2.5)--(4+0.9,-4.5);
\draw (4,-2.5)--(4,-4.5);
\draw (4,-2.5)--(4-0.9,-4.5);
\draw[draw,line width=1mm] (-6-0.9-0.4,-6)--(-6-0.9,-4.5);
\draw (-6-0.9+0.2,-6)--(-6-0.9,-4.5);
\draw (-6-0.2-0.3,-6)--(-6-0.2,-4.5);
\draw (-6-0.2,-6)--(-6-0.2,-4.5);
\draw (-6-0.2+0.3,-6)--(-6-0.2,-4.5);
\draw (-6+0.5-0.2,-6)--(-6+0.5,-4.5);
\draw (-6+0.5+0.2,-6)--(-6+0.5,-4.5);
\draw (-4-0.9-0.2,-6)--(-4-0.9,-4.5);
\draw (-4-0.9+0.1,-6)--(-4-0.9,-4.5);
\draw (-4-0.65,-6)--(-4-0.45,-4.5);
\draw (-4-0.45,-6)--(-4-0.45,-4.5);
\draw (-4-0.25,-6)--(-4-0.45,-4.5);
\draw (-4-0.1,-6)--(-4,-4.5);
\draw (-4+0.1,-6)--(-4,-4.5);
\draw (-4+0.65,-6)--(-4+0.45,-4.5);
\draw (-4+0.45,-6)--(-4+0.45,-4.5);
\draw (-4+0.25,-6)--(-4+0.45,-4.5);
\draw (-4+0.9-0.1,-6)--(-4+0.9,-4.5);
\draw (-4+0.9+0.2,-6)--(-4+0.9,-4.5);
\draw (-2-0.5-0.2,-6)--(-2-0.5,-4.5);
\draw (-2-0.5+0.2,-6)--(-2-0.5,-4.5);
\draw (-2+0.2-0.3,-6)--(-2+0.2,-4.5);
\draw (-2+0.2,-6)--(-2+0.2,-4.5);
\draw (-2+0.2+0.3,-6)--(-2+0.2,-4.5);
\draw (-2+0.9-0.2,-6)--(-2+0.9,-4.5);
\draw (-2+0.9+0.4,-6)--(-2+0.9,-4.5);
\draw (8-6-0.9-0.4,-6)--(8-6-0.9,-4.5);
\draw (8-6-0.9+0.2,-6)--(8-6-0.9,-4.5);
\draw (8-6-0.2-0.3,-6)--(8-6-0.2,-4.5);
\draw (8-6-0.2,-6)--(8-6-0.2,-4.5);
\draw (8-6-0.2+0.3,-6)--(8-6-0.2,-4.5);
\draw (8-6+0.5-0.2,-6)--(8-6+0.5,-4.5);
\draw (8-6+0.5+0.2,-6)--(8-6+0.5,-4.5);
\draw (8-4-0.9-0.2,-6)--(8-4-0.9,-4.5);
\draw (8-4-0.9+0.1,-6)--(8-4-0.9,-4.5);
\draw (8-4-0.65,-6)--(8-4-0.45,-4.5);
\draw (8-4-0.45,-6)--(8-4-0.45,-4.5);
\draw (8-4-0.25,-6)--(8-4-0.45,-4.5);
\draw (8-4-0.1,-6)--(8-4,-4.5);
\draw (8-4+0.1,-6)--(8-4,-4.5);
\draw (8-4+0.65,-6)--(8-4+0.45,-4.5);
\draw (8-4+0.45,-6)--(8-4+0.45,-4.5);
\draw (8-4+0.25,-6)--(8-4+0.45,-4.5);
\draw (8-4+0.9-0.1,-6)--(8-4+0.9,-4.5);
\draw (8-4+0.9+0.2,-6)--(8-4+0.9,-4.5);
\draw (8-2-0.5-0.2,-6)--(8-2-0.5,-4.5);
\draw (8-2-0.5+0.2,-6)--(8-2-0.5,-4.5);
\draw (8-2+0.2-0.3,-6)--(8-2+0.2,-4.5);
\draw (8-2+0.2,-6)--(8-2+0.2,-4.5);
\draw (8-2+0.2+0.3,-6)--(8-2+0.2,-4.5);
\draw (8-2+0.9-0.2,-6)--(8-2+0.9,-4.5);
\draw[draw,line width=1mm] (8-2+0.9+0.4,-6)--(8-2+0.9,-4.5);
\end{tikzpicture}
\vspace{3mm}

\subcaption{An entanglement wedge at a finite cutoff.}
\end{subfigure}
\vspace{8mm}

\begin{subfigure}{1\textwidth}
\begin{tikzpicture}
\draw[color=red] (4,0)--(5,1);
\draw[color=red] (-4,0)--(-5,1);
\draw[color=red] (-8,0)--(8,0);
\draw[color=blue] (-8,-2.5)--(-2,-2.5);
\draw[color=blue] (8,-2.5)--(2,-2.5);
\draw[color=red] (-2,-2.5)--(2,-2.5);
\draw (-8,-4.5)--(-2+0.9,-4.5);
\draw (8,-4.5)--(2-0.9,-4.5);
\draw[color=red] (-2+0.9,-4.5)--(2-0.9,-4.5);
\draw[color=red] (-4,0)--(-2,-2.5);
\draw[color=red] (2,-2.5)--(4,0);
\draw[color=red] (4,0)--(6,-2.5);
\draw[color=red] (-6,-2.5)--(-4,0);
\draw[color=blue] (-6,-2.5)--(-6-0.9,-4.5);
\draw[color=blue] (-6,-2.5)--(-6+0.5,-4.5);
\draw[color=red] (-2+0.9,-4.5)--(-2,-2.5);
\draw[color=red] (-2-0.5,-4.5)--(-2,-2.5);
\draw[color=blue] (-4,-2.5)--(-4+0.9,-4.5);
\draw[color=blue] (-4,-2.5)--(-4,-4.5);
\draw[color=blue] (-4,-2.5)--(-4-0.9,-4.5);
\draw[color=blue] (6,-2.5)--(6+0.9,-4.5);
\draw[color=blue] (6,-2.5)--(6-0.5,-4.5);
\draw[color=red] (2+0.5,-4.5)--(2,-2.5);
\draw[color=red] (2-0.9,-4.5)--(2,-2.5);
\draw[color=blue] (4,-2.5)--(4+0.9,-4.5);
\draw[color=blue] (4,-2.5)--(4,-4.5);
\draw[color=blue] (4,-2.5)--(4-0.9,-4.5);
\draw (-6-0.9-0.4,-6)--(-6-0.9,-4.5);
\draw (-6-0.9+0.2,-6)--(-6-0.9,-4.5);
\draw (-6-0.2-0.3,-6)--(-6-0.2,-4.5);
\draw (-6-0.2,-6)--(-6-0.2,-4.5);
\draw (-6-0.2+0.3,-6)--(-6-0.2,-4.5);
\draw (-6+0.5-0.2,-6)--(-6+0.5,-4.5);
\draw (-6+0.5+0.2,-6)--(-6+0.5,-4.5);
\draw (-4-0.9-0.2,-6)--(-4-0.9,-4.5);
\draw (-4-0.9+0.1,-6)--(-4-0.9,-4.5);
\draw (-4-0.65,-6)--(-4-0.45,-4.5);
\draw (-4-0.45,-6)--(-4-0.45,-4.5);
\draw (-4-0.25,-6)--(-4-0.45,-4.5);
\draw (-4-0.1,-6)--(-4,-4.5);
\draw (-4+0.1,-6)--(-4,-4.5);
\draw (-4+0.65,-6)--(-4+0.45,-4.5);
\draw (-4+0.45,-6)--(-4+0.45,-4.5);
\draw (-4+0.25,-6)--(-4+0.45,-4.5);
\draw (-4+0.9-0.1,-6)--(-4+0.9,-4.5);
\draw (-4+0.9+0.2,-6)--(-4+0.9,-4.5);
\draw (-2-0.5-0.2,-6)--(-2-0.5,-4.5);
\draw (-2-0.5+0.2,-6)--(-2-0.5,-4.5);
\draw (-2+0.2-0.3,-6)--(-2+0.2,-4.5);
\draw (-2+0.2,-6)--(-2+0.2,-4.5);
\draw (-2+0.2+0.3,-6)--(-2+0.2,-4.5);
\draw[color=red] (-2+0.9-0.2,-6)--(-2+0.9,-4.5);
\draw[color=red] (-2+0.9+0.4,-6)--(-2+0.9,-4.5);
\draw[color=red] (8-6-0.9-0.4,-6)--(8-6-0.9,-4.5);
\draw[color=red] (8-6-0.9+0.2,-6)--(8-6-0.9,-4.5);
\draw (8-6-0.2-0.3,-6)--(8-6-0.2,-4.5);
\draw (8-6-0.2,-6)--(8-6-0.2,-4.5);
\draw (8-6-0.2+0.3,-6)--(8-6-0.2,-4.5);
\draw (8-6+0.5-0.2,-6)--(8-6+0.5,-4.5);
\draw (8-6+0.5+0.2,-6)--(8-6+0.5,-4.5);
\draw (8-4-0.9-0.2,-6)--(8-4-0.9,-4.5);
\draw (8-4-0.9+0.1,-6)--(8-4-0.9,-4.5);
\draw (8-4-0.65,-6)--(8-4-0.45,-4.5);
\draw (8-4-0.45,-6)--(8-4-0.45,-4.5);
\draw (8-4-0.25,-6)--(8-4-0.45,-4.5);
\draw (8-4-0.1,-6)--(8-4,-4.5);
\draw (8-4+0.1,-6)--(8-4,-4.5);
\draw (8-4+0.65,-6)--(8-4+0.45,-4.5);
\draw (8-4+0.45,-6)--(8-4+0.45,-4.5);
\draw (8-4+0.25,-6)--(8-4+0.45,-4.5);
\draw (8-4+0.9-0.1,-6)--(8-4+0.9,-4.5);
\draw (8-4+0.9+0.2,-6)--(8-4+0.9,-4.5);
\draw (8-2-0.5-0.2,-6)--(8-2-0.5,-4.5);
\draw (8-2-0.5+0.2,-6)--(8-2-0.5,-4.5);
\draw (8-2+0.2-0.3,-6)--(8-2+0.2,-4.5);
\draw (8-2+0.2,-6)--(8-2+0.2,-4.5);
\draw (8-2+0.2+0.3,-6)--(8-2+0.2,-4.5);
\draw (8-2+0.9-0.2,-6)--(8-2+0.9,-4.5);
\draw (8-2+0.9+0.4,-6)--(8-2+0.9,-4.5);
\end{tikzpicture}
\vspace{3mm}

\subcaption{The reverse engineering of the HaPPY code.}
\end{subfigure}
\vspace{3mm}

\caption{The straight-forward HaPPY code from the center and its reverse engineering from the boundary for the infinite-dimensional analog of the HaPPY code.\label{reverse}}
\end{figure}
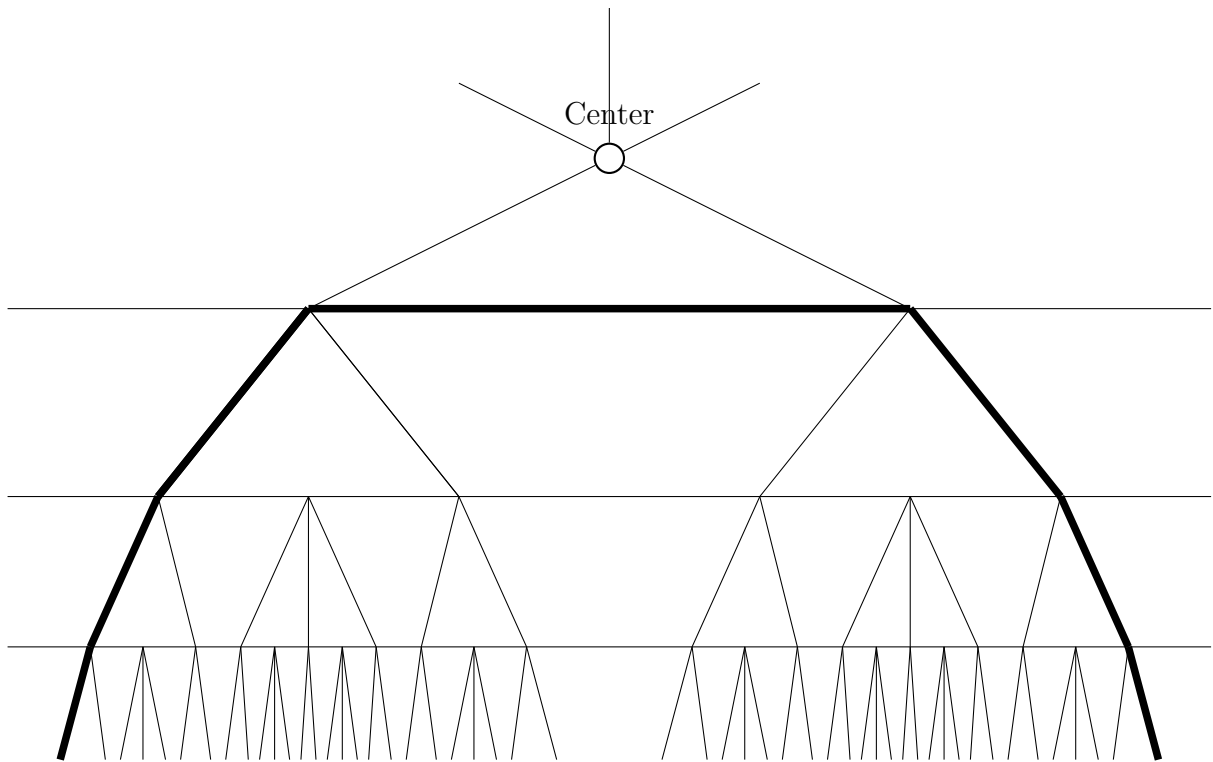
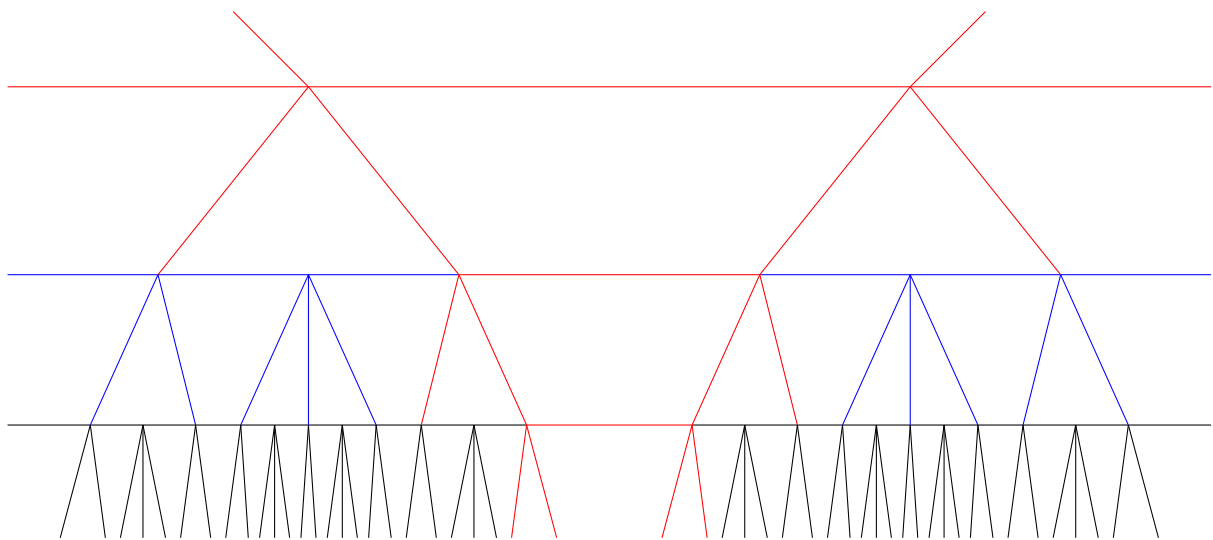

\subsection{Bulk and boundary $C^*$-algebras and isometric mapping}
In order to implement Theorem \ref{thm:maintheoremCstar} to the infinite-dimensional HaPPY code, we now define our setup more formally. Consider an infinite string of qubits indexed by the integers, and construct an infinite graph in the upper half-plane by piling up trapezes made of 2-clusters on top of each other, and constructing the graph underneath by adding 2-clusters on the bottom of the legs and 3-clusters in the middle of the 2-clusters. Associate to each bulk and boundary node the spin $1/2$ algebra $M_2(\mathbb{C})$. Define a $\star$-algebra $\mathcal{F}_{code}$ consisting of operators with finite support on the bulk graph, and a $\star$-algebra $\mathcal{F}_{phys}$ consisting of operators with finite support on the boundary string of qubits. Operator pushing defines an isometric map $\iota$ from $\mathcal{F}_{code}$ to $\mathcal{F}_{phys}$.

Taking the norm completion of $\mathcal{F}_{code}$ and $\mathcal{F}_{phys}$ gives the $C^*$-algebras $\mathcal{A}_{code}$ and $\mathcal{A}_{phys}$. As $\iota$ preserves the norm, it sends Cauchy sequences onto Cauchy sequences, which shows that it extends to an isometric $C^*$-homomorphism from $\mathcal{A}_{code}$ to $\mathcal{A}_{phys}$.\footnote{Note that this construction is the standard construction for uniformly hyperfinite (UHF) algebras.}

\subsection{Trapeze dynamics and KMS states}
We now need to define dynamics on our $C^*$-algebra, and prove that it admits KMS states. There is of course many possible choices, and we expect a lot of them to work. For the sake of convenience, we inspire our choice here from the dynamical model extensively studied in \cite{MonicaElliott}.

As we did similarly in \cite{MonicaElliott}, we construct the dynamics on the HaPPY code out of the trapeze Hamiltonian. For each trapeze shape in the bulk, we define the interaction term
\begin{align}
\prod_{i\in\mathrm{trapeze}}Z_i,
\end{align}
where $Z$ denotes the spin Pauli matrix. As shown on Figure \ref{fig:Trepeze1stBD}, a trapeze interaction maps nicely to the first boundary layer. It is then pushed to the boundary through a Cantor-like fractal pattern. In our reverse engineering setting, all bulk trapezes are pushed to the boundary after a finite number of operations. The result of the Appendix of \cite{MonicaElliott} shows that it is possible to map a carefully engineered trapeze Hamiltonian in the bulk to a boundary operator which only correlates finite families of boundary qubits together. The physical implications of this mapping are discussed in \cite{MonicaElliott}. We shall now prove that the obtained boundary Hamiltonian allows us to rigorously define a strongly continuous one-parameter group of isometries on $\mathcal{A}_{phys}$.

\begin{figure}[H]
\begin{center}
\begin{tikzpicture}
\draw (-4,0)--(-4,1);
\draw (4,0)--(4,1);
\draw (-8,0)--(8,0);
\draw (-8,-2.5)--(8,-2.5);
\draw (-4,0)--(4,0);
\draw (-6,-2.5)--(-4,0);
\draw (-4,0)--(-2,-2.5);
\draw (2,-2.5)--(4,0);
\draw (4,0)--(6,-2.5);
\draw (-6,-2.5)--(-4,0);
\draw (-4,0)--(-2,-2.5);
\draw (-6,-2.5)--(-6-0.9,-4.5);
\draw (-6,-2.5)--(-6+0.9,-4.5);
\draw (-2+0.9,-4.5)--(-2,-2.5);
\draw (-2-0.9,-4.5)--(-2,-2.5);
\draw (-4,-2.5)--(-4+0.9,-4.5);
\draw (-4,-2.5)--(-4,-4.5);
\draw (-4,-2.5)--(-4-0.9,-4.5);
\draw (6,-2.5)--(6+0.9,-4.5);
\draw (6,-2.5)--(6-0.9,-4.5);
\draw (2+0.9,-4.5)--(2,-2.5);
\draw (2-0.9,-4.5)--(2,-2.5);
\draw (4,-2.5)--(4+0.9,-4.5);
\draw (4,-2.5)--(4,-4.5);
\draw (4,-2.5)--(4-0.9,-4.5);
\node[draw,circle,thick,scale=1,black,label={[label distance=1mm]south:Z}] (C) at (-4,0) {};
\node[draw,circle,thick,scale=1,black,label={[label distance=1mm]south:Z}] (C) at (4,0) {};
\node[draw,circle,thick,scale=1,black,label={[label distance=1mm]north:Z}] (C) at (-6,-2.5) {};
\node[draw,circle,thick,scale=1,black,label={[label distance=1mm]north:Z}] (C) at (-2,-2.5) {};
\node[draw,circle,thick,scale=1,black,label={[label distance=1mm]north:Z}] (C) at (2,-2.5) {};
\node[draw,circle,thick,scale=1,black,label={[label distance=1mm]north:Z}] (C) at (6,-2.5) {};
\node[draw,circle,thick,scale=1,black,label={[label distance=1mm]north:Z}] (C) at (-4,-2.5) {};
\node[draw,circle,thick,scale=1,black,label={[label distance=1mm]north:Z}] (C) at (4,-2.5) {};
\node[draw=none,label={[label distance=0.5mm]south:YY}] (C) at (-5,-4.5) {};
\node[draw=none,label={[label distance=0.5mm]south:YY}] (C) at (5,-4.5) {};
\end{tikzpicture}
\caption{Pushing a trapeze operator to the first boundary layer.}
\label{fig:Trepeze1stBD}
\end{center}
\end{figure}
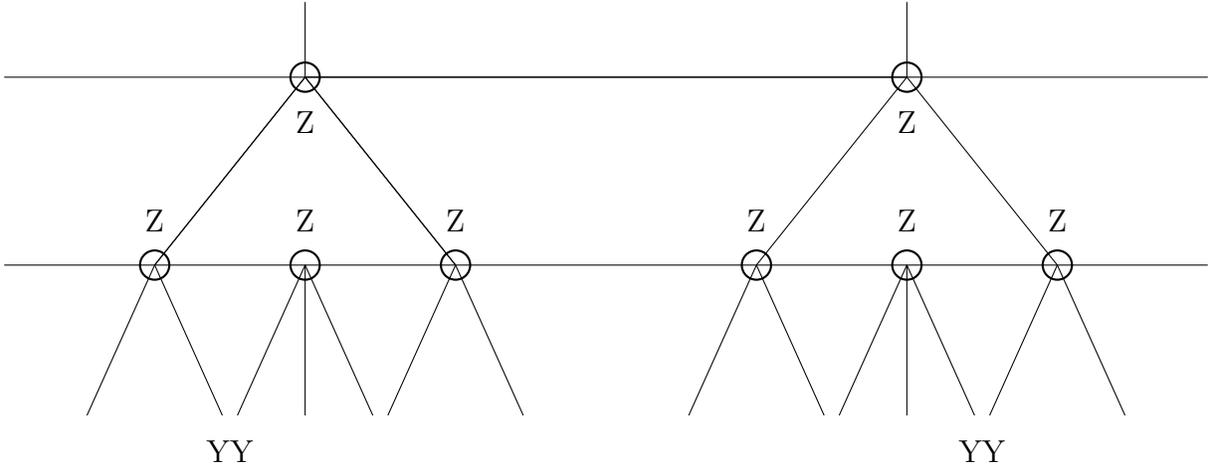

For $n\in\mathbb{N}$, we first define the truncated time evolution $\sigma_t^n$ on $\mathcal{F}_{phys}$ by 
\begin{align}
\sigma_t^n(A):=e^{iH_nt}Ae^{-iH_nt},
\end{align}
where $H_n$ is the operator made of all trapeze pushes with support in $\{-n,...,n\}$. Since it is proven in \cite[Appendix A]{MonicaElliott} that qubits in $\{-n,...,n\}$ are coupled to only a finite number of other qubits on the boundary by all the trapeze pushes, we deduce that for $A\in\mathcal{F}_{phys}$, the sequence $(\sigma_t^n(A))$ is stationary. It therefore makes sense to define on $\mathcal{F}_{phys}$ the time evolution
\begin{align}
\sigma_t(A):=\underset{n\rightarrow\infty}{\mathrm{lim}}\sigma_t^n(A).
\end{align}
This time evolution easily extends to $\mathcal{A}_{phys}$ as it maps Cauchy sequences to Cauchy sequences. Note that, as the time evolutions $\sigma_t^n$ are constructed out of Hamiltonians of the form $\iota(K_n)$, where  $K_n\in\mathcal{A}_{bulk}$, and the $\iota$ is the operator-pushing $C^*$-homomorphism, each $\sigma_t^n$ stabilizes $\iota(\mathcal{A}_{code})$, which implies by construction of $\sigma_t$ that $\sigma_t$ stabilizes $\iota(\mathcal{F}_{code})$, and therefore $\mathcal{A}_{code}$ by going to the limit, since $\iota$ is an isometry. Moreover, $\sigma_t$ is strongly (and even uniformly) continuous. Let $\varepsilon>0$ and for $A\in\mathcal{A}_{phys}$, let ${A_f}\in\mathcal{F}_{phys}$ such that 
\begin{align}
\|A-A_f\|\leq\frac{\varepsilon}{3}.
\end{align}
Then, there clearly exists $t_0>0$ such that for $0\leq t\leq t_0$,
\begin{align}
\|\sigma_t(A_f)-A_f\|\leq\frac{\varepsilon}{3}.
\end{align}
Then, we find that 
\begin{align}
\|\sigma_t(A)-A\|\leq\|\sigma_t(A)-\sigma_t(A_f)\|+\|\sigma_t(A_f)-A_f\|+\|A_f-A\|=\frac{3\varepsilon}{3}=\varepsilon.
\end{align}
We have defined a strongly continuous time evolution $\sigma_t$ on $\mathcal{A}_{phys}$. It now remains to show that for $\beta>0$, there exists a KMS state for $\sigma_t$ at inverse temperature $\beta$. Here, we give a simplified version of a standard argument from the theory of infinite-dimensional quantum spin systems, which can for instance be found in \cite{Pillet}.

For fixed inverse temperature $\beta$, one can consider the state on $\mathcal{F}_{phys}$ to be
\begin{align}
\omega_{n}(A):=\frac{\mathrm{Tr}_{n,A}(Ae^{-\beta H_n})}{\mathrm{Tr}_{n,A}(e^{-\beta H_n})},
\end{align}
where $H_n$ is defined as before and the trace $\mathrm{Tr}_{n,A}$ is taken over the union of the support of $H_n$ and the support of $A$. For $A\in\mathcal{F}_{phys}$, the sequence $(\omega_n(A))$ is stationary as the pushes of all trapeze operators only couple finite numbers of boundary qubits together. It therefore makes sense to define the state 
\begin{align}
\omega(A):=\underset{n\rightarrow\infty}{\mathrm{lim}}\omega_n(A)
\end{align}
for $A\in\mathcal{F}_{phys}$. As this sequence is stationary, this state can be identified with a Gibbs state when applied to any given $A,B$ in $\mathcal{F}_{phys}$, and will therefore satisfy the KMS condition on $\mathcal{F}_{phys}$. We are left with proving that this can be extended to a KMS state on the whole $\mathcal{A_{phys}}$, which requires complex analysis.

Let $A, B\in\mathcal{A}_{phys}$ and let $(A_n)$ and $(B_n)$ be sequences of $\mathcal{F}_{phys}$ which converge in norm towards $A$ and $B$ respectively. By the KMS condition, for all $n\in\mathbb{N}$, there exists a function $F_{A_n,B_n}$, analytic on the strip $\{0<\mathrm{Im}z<\beta\}$ and continuous on its closure, such that
\begin{align}
F_{A_n,B_n}(t)=\omega(A_n\sigma_t(B_n))\quad\text{and}\quad F_{A_n,B_n}(t+i\beta)=\omega(\sigma_t(A_n)B_n).
\end{align}
By using Hadamard's three-line theorem \cite{Hadamard}, we find the bound to be
\begin{align}
\underset{\mathrm{Im}z\in [0,\beta]}{\mathrm{sup}}|F_{A_n,B_n}(z)|\leq\|A_n\|\|B_n\|.
\end{align}
Then we notice that 
\begin{align}
F_{A_n,B_n}(z)-F_{A_m,B_m}(z)=F_{A_n-A_m,B_n}(z)+F_{A_n,B_n-B_m}(z)
\end{align}
and that the sequences $(A_n)$ and $(B_n)$ are Cauchy, the former bound allows us to conclude that the sequence $(F_{A_n,B_n})$ is uniformly convergent in the strip $\{\mathrm{Im}z\in [0,\beta]\}$. Its uniform limit is holomorphic on the open strip, continuous on the closed strip, and satisfies the KMS conditions at inverse temperature $\beta$ for the state $\omega$, which ends the proof that $\omega$ is a KMS state on $\mathcal{A}_{phys}$ for the time evolution $\sigma_t$.

\subsection{The holographic HaPPY wormhole}
We have constructed two $C^*$-algebras $\mathcal{A}_{code}$ and $\mathcal{A}_{phys}$, and an isometric map $\iota$ between them. Moreover, for the strongly continuous one-parameter group of isometries $\sigma_t$ that we constructed on $\mathcal{A}_{phys}$, $\iota(\mathcal{A}_{code})$ is stabilized and there exists a KMS state at every inverse temperature. We can therefore use our main result to conclude that there exist Hilbert space representations of $\mathcal{A}_{phys}$ and $\mathcal{A}_{code}$, and a Hilbert space isometry between them, which satisfies bulk reconstruction, conserves relative entropy and preserves modular flow.

This proof shows the power of our result for holographic error correcting codes which, like the HaPPY code, are better-suited to operator pushing than state pushing. Through only the construction of the $C^*$-algebras of observables and reasonable dynamics, we were able to construct thermofield double Hilbert spaces, show bulk reconstruction at the level of the states, and prove the conservation of relative entropies and modular flow between the bulk and the boundary. Moreover, Tomita-Takesaki theory provides us with a copy of the involved von Neumann algebras, which are interpreted as local algebras on the other side of the spacelike slice represented by our system. We have constructed an infinite-dimensional HaPPY wormhole! 

Figure \ref{fig:HaPPYwormhole} shows the two-sided holographic HaPPY wormhole that we constructed. One of the sides is acted on by the von Neumann algebras $M_{code}$ and $M_{phys}$, which respectively correspond to bulk and boundary quasi-local observables. The other side is acted on by the commutant algebras $M^\prime_{code}$ and $M^\prime_{phys}$, defined with respect to the Hilbert spaces $\mathcal{H}_{code}$ and $\mathcal{H}_{phys}$ through Tomita-Takesaki theory. Our wormhole satisfies all expected properties of entanglement wedge reconstruction.

As we discussed before, this construction never really reaches the black hole interior, from either side of the wormhole. However, it would be interesting to see if the Papadodimas-Raju proposal can be applied to approximately construct a black hole interior, by acting on both boundaries at the same time. \\

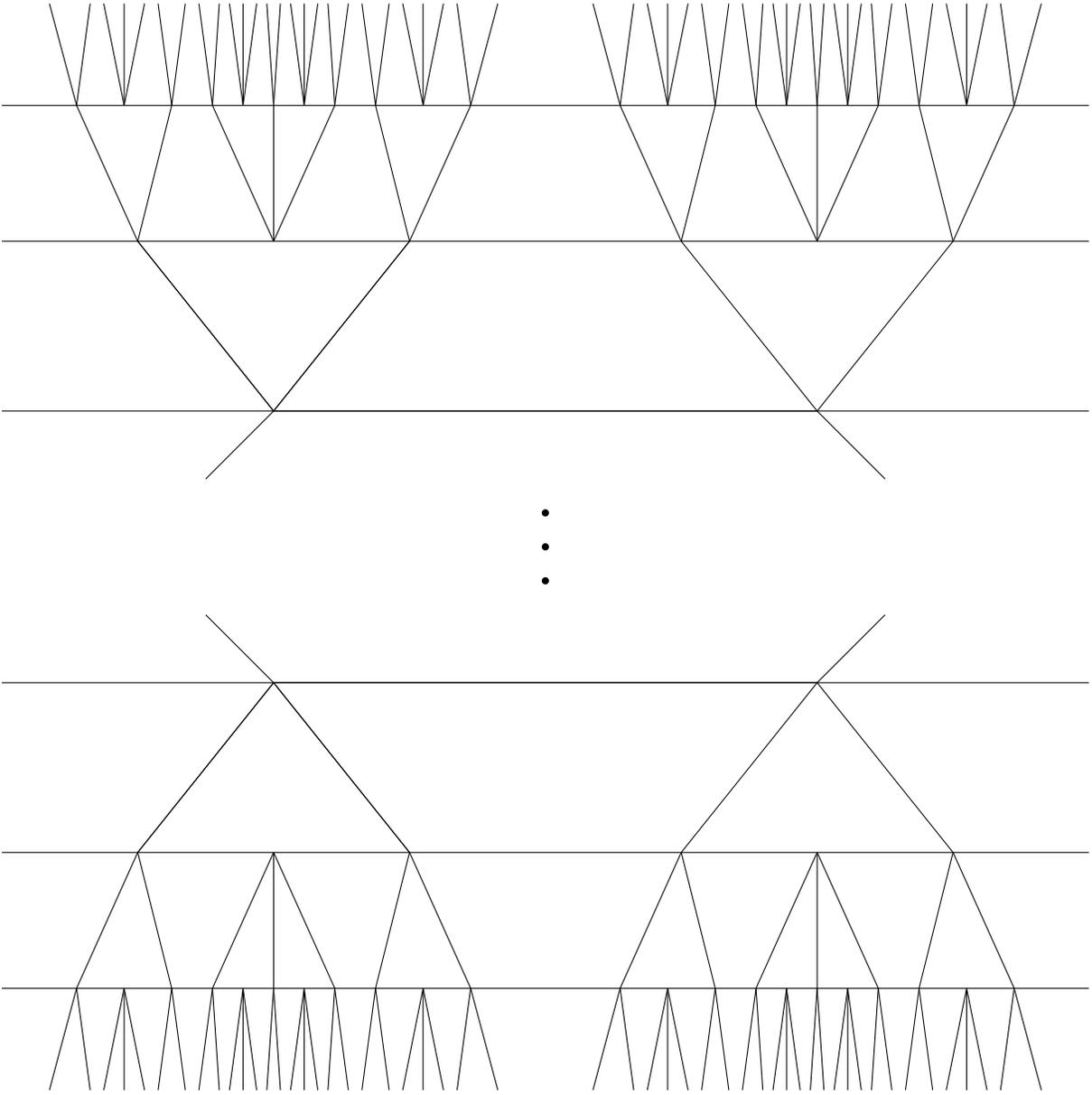
\begin{figure}[H]
\begin{center}
\begin{tikzpicture}
\draw (-8,0)--(8,0);
\draw (-8,-2.5)--(8,-2.5);
\draw (-8,-4.5)--(8,-4.5);
\draw (-4,0)--(4,0);
\draw (-6,-2.5)--(-4,0);
\draw (-4,0)--(-2,-2.5);
\draw (2,-2.5)--(4,0);
\draw (4,0)--(6,-2.5);
\draw (-6,-2.5)--(-4,0);
\draw (-4,0)--(-2,-2.5);
\draw (-6,-2.5)--(-6-0.9,-4.5);
\draw (-6,-2.5)--(-6+0.5,-4.5);
\draw (-2+0.9,-4.5)--(-2,-2.5);
\draw (-2-0.5,-4.5)--(-2,-2.5);
\draw (-4,-2.5)--(-4+0.9,-4.5);
\draw (-4,-2.5)--(-4,-4.5);
\draw (-4,-2.5)--(-4-0.9,-4.5);
\draw (6,-2.5)--(6+0.9,-4.5);
\draw (6,-2.5)--(6-0.5,-4.5);
\draw (2+0.5,-4.5)--(2,-2.5);
\draw (2-0.9,-4.5)--(2,-2.5);
\draw (4,-2.5)--(4+0.9,-4.5);
\draw (4,-2.5)--(4,-4.5);
\draw (4,-2.5)--(4-0.9,-4.5);
\draw (-6-0.9-0.4,-6)--(-6-0.9,-4.5);
\draw (-6-0.9+0.2,-6)--(-6-0.9,-4.5);
\draw (-6-0.2-0.3,-6)--(-6-0.2,-4.5);
\draw (-6-0.2,-6)--(-6-0.2,-4.5);
\draw (-6-0.2+0.3,-6)--(-6-0.2,-4.5);
\draw (-6+0.5-0.2,-6)--(-6+0.5,-4.5);
\draw (-6+0.5+0.2,-6)--(-6+0.5,-4.5);
\draw (-4-0.9-0.2,-6)--(-4-0.9,-4.5);
\draw (-4-0.9+0.1,-6)--(-4-0.9,-4.5);
\draw (-4-0.65,-6)--(-4-0.45,-4.5);
\draw (-4-0.45,-6)--(-4-0.45,-4.5);
\draw (-4-0.25,-6)--(-4-0.45,-4.5);
\draw (-4-0.1,-6)--(-4,-4.5);
\draw (-4+0.1,-6)--(-4,-4.5);
\draw (-4+0.65,-6)--(-4+0.45,-4.5);
\draw (-4+0.45,-6)--(-4+0.45,-4.5);
\draw (-4+0.25,-6)--(-4+0.45,-4.5);
\draw (-4+0.9-0.1,-6)--(-4+0.9,-4.5);
\draw (-4+0.9+0.2,-6)--(-4+0.9,-4.5);
\draw (-2-0.5-0.2,-6)--(-2-0.5,-4.5);
\draw (-2-0.5+0.2,-6)--(-2-0.5,-4.5);
\draw (-2+0.2-0.3,-6)--(-2+0.2,-4.5);
\draw (-2+0.2,-6)--(-2+0.2,-4.5);
\draw (-2+0.2+0.3,-6)--(-2+0.2,-4.5);
\draw (-2+0.9-0.2,-6)--(-2+0.9,-4.5);
\draw (-2+0.9+0.4,-6)--(-2+0.9,-4.5);
\draw (8-6-0.9-0.4,-6)--(8-6-0.9,-4.5);
\draw (8-6-0.9+0.2,-6)--(8-6-0.9,-4.5);
\draw (8-6-0.2-0.3,-6)--(8-6-0.2,-4.5);
\draw (8-6-0.2,-6)--(8-6-0.2,-4.5);
\draw (8-6-0.2+0.3,-6)--(8-6-0.2,-4.5);
\draw (8-6+0.5-0.2,-6)--(8-6+0.5,-4.5);
\draw (8-6+0.5+0.2,-6)--(8-6+0.5,-4.5);
\draw (8-4-0.9-0.2,-6)--(8-4-0.9,-4.5);
\draw (8-4-0.9+0.1,-6)--(8-4-0.9,-4.5);
\draw (8-4-0.65,-6)--(8-4-0.45,-4.5);
\draw (8-4-0.45,-6)--(8-4-0.45,-4.5);
\draw (8-4-0.25,-6)--(8-4-0.45,-4.5);
\draw (8-4-0.1,-6)--(8-4,-4.5);
\draw (8-4+0.1,-6)--(8-4,-4.5);
\draw (8-4+0.65,-6)--(8-4+0.45,-4.5);
\draw (8-4+0.45,-6)--(8-4+0.45,-4.5);
\draw (8-4+0.25,-6)--(8-4+0.45,-4.5);
\draw (8-4+0.9-0.1,-6)--(8-4+0.9,-4.5);
\draw (8-4+0.9+0.2,-6)--(8-4+0.9,-4.5);
\draw (8-2-0.5-0.2,-6)--(8-2-0.5,-4.5);
\draw (8-2-0.5+0.2,-6)--(8-2-0.5,-4.5);
\draw (8-2+0.2-0.3,-6)--(8-2+0.2,-4.5);
\draw (8-2+0.2,-6)--(8-2+0.2,-4.5);
\draw (8-2+0.2+0.3,-6)--(8-2+0.2,-4.5);
\draw (8-2+0.9-0.2,-6)--(8-2+0.9,-4.5);
\draw (8-2+0.9+0.4,-6)--(8-2+0.9,-4.5);
\draw (-4,4)--(-5,3);
\draw (4,4)--(5,3);
\draw (-4,0)--(-5,1);
\draw (4,0)--(5,1);
\draw (-8,4)--(8,4);
\draw (-8,6.5)--(8,6.5);
\draw (-8,8.5)--(8,8.5);
\draw (-4,4)--(4,4);
\draw (-6,6.5)--(-4,4);
\draw (-4,4)--(-2,6.5);
\draw (2,6.5)--(4,4);
\draw (4,4)--(6,6.5);
\draw (-6,6.5)--(-4,4);
\draw (-4,4)--(-2,6.5);
\draw (-6,6.5)--(-6-0.9,8.5);
\draw (-6,6.5)--(-6+0.5,8.5);
\draw (-2+0.9,8.5)--(-2,6.5);
\draw (-2-0.5,8.5)--(-2,6.5);
\draw (-4,6.5)--(-4+0.9,8.5);
\draw (-4,6.5)--(-4,8.5);
\draw (-4,6.5)--(-4-0.9,8.5);
\draw (6,6.5)--(6+0.9,8.5);
\draw (6,6.5)--(6-0.5,8.5);
\draw (2+0.5,8.5)--(2,6.5);
\draw (2-0.9,8.5)--(2,6.5);
\draw (4,6.5)--(4+0.9,8.5);
\draw (4,6.5)--(4,8.5);
\draw (4,6.5)--(4-0.9,8.5);
\draw (-6-0.9-0.4,10)--(-6-0.9,8.5);
\draw (-6-0.9+0.2,10)--(-6-0.9,8.5);
\draw (-6-0.2-0.3,10)--(-6-0.2,8.5);
\draw (-6-0.2,10)--(-6-0.2,8.5);
\draw (-6-0.2+0.3,10)--(-6-0.2,8.5);
\draw (-6+0.5-0.2,10)--(-6+0.5,8.5);
\draw (-6+0.5+0.2,10)--(-6+0.5,8.5);
\draw (-4-0.9-0.2,10)--(-4-0.9,8.5);
\draw (-4-0.9+0.1,10)--(-4-0.9,8.5);
\draw (-4-0.65,10)--(-4-0.45,8.5);
\draw (-4-0.45,10)--(-4-0.45,8.5);
\draw (-4-0.25,10)--(-4-0.45,8.5);
\draw (-4-0.1,10)--(-4,8.5);
\draw (-4+0.1,10)--(-4,8.5);
\draw (-4+0.65,10)--(-4+0.45,8.5);
\draw (-4+0.45,10)--(-4+0.45,8.5);
\draw (-4+0.25,10)--(-4+0.45,8.5);
\draw (-4+0.9-0.1,10)--(-4+0.9,8.5);
\draw (-4+0.9+0.2,10)--(-4+0.9,8.5);
\draw (-2-0.5-0.2,10)--(-2-0.5,8.5);
\draw (-2-0.5+0.2,10)--(-2-0.5,8.5);
\draw (-2+0.2-0.3,10)--(-2+0.2,8.5);
\draw (-2+0.2,10)--(-2+0.2,8.5);
\draw (-2+0.2+0.3,10)--(-2+0.2,8.5);
\draw (-2+0.9-0.2,10)--(-2+0.9,8.5);
\draw (-2+0.9+0.4,10)--(-2+0.9,8.5);
\draw (8-6-0.9-0.4,10)--(8-6-0.9,8.5);
\draw (8-6-0.9+0.2,10)--(8-6-0.9,8.5);
\draw (8-6-0.2-0.3,10)--(8-6-0.2,8.5);
\draw (8-6-0.2,10)--(8-6-0.2,8.5);
\draw (8-6-0.2+0.3,10)--(8-6-0.2,8.5);
\draw (8-6+0.5-0.2,10)--(8-6+0.5,8.5);
\draw (8-6+0.5+0.2,10)--(8-6+0.5,8.5);
\draw (8-4-0.9-0.2,10)--(8-4-0.9,8.5);
\draw (8-4-0.9+0.1,10)--(8-4-0.9,8.5);
\draw (8-4-0.65,10)--(8-4-0.45,8.5);
\draw (8-4-0.45,10)--(8-4-0.45,8.5);
\draw (8-4-0.25,10)--(8-4-0.45,8.5);
\draw (8-4-0.1,10)--(8-4,8.5);
\draw (8-4+0.1,10)--(8-4,8.5);
\draw (8-4+0.65,10)--(8-4+0.45,8.5);
\draw (8-4+0.45,10)--(8-4+0.45,8.5);
\draw (8-4+0.25,10)--(8-4+0.45,8.5);
\draw (8-4+0.9-0.1,10)--(8-4+0.9,8.5);
\draw (8-4+0.9+0.2,10)--(8-4+0.9,8.5);
\draw (8-2-0.5-0.2,10)--(8-2-0.5,8.5);
\draw (8-2-0.5+0.2,10)--(8-2-0.5,8.5);
\draw (8-2+0.2-0.3,10)--(8-2+0.2,8.5);
\draw (8-2+0.2,10)--(8-2+0.2,8.5);
\draw (8-2+0.2+0.3,10)--(8-2+0.2,8.5);
\draw (8-2+0.9-0.2,10)--(8-2+0.9,8.5);
\draw (8-2+0.9+0.4,10)--(8-2+0.9,8.5);
\node[draw,circle,thick,scale=0.2,black,fill=black] (L1) at (0,1.5) {};
\node[draw,circle,thick,scale=0.2,black,fill=black] (L1) at (0,2) {};
\node[draw,circle,thick,scale=0.2,black,fill=black] (L1) at (0,2.5) {};
\end{tikzpicture}
\vspace{5mm}
\caption{The holographic HaPPY wormhole with its two infinite boundaries.}
\label{fig:HaPPYwormhole}
\end{center}
\end{figure}

\section{Discussion} \label{sec:discussion}
In this paper, we introduced a new result on holographic quantum error correction in infinite dimensions. With the only given data of a bulk-to-boundary isometric $C^*$-homomorphism, and dynamics on the boundary for which a KMS state can be defined, we constructed Hilbert space representations of the bulk and boundary algebras for which a bulk-to-boundary map can be constructed at the level of the states, and relative entropies and modular time are conserved between the bulk and the boundary. Our construction is very close in spirit to the thermofield double construction in AdS/CFT, and should be seen as the construction of the other boundary of a wormhole. In particular, it clarifies the claims of Papadodimas and Raju in \cite{Papadodimas:2013wnh,Papadodimas:2013jku,Papadodimas:2012aq}. We showed that our theorem can be applied to construct an infinite-dimensional HaPPY wormhole. Its main strength is that it is adapted to systems defined in terms of operator pushing rather than state pushing. In particular, it is very well-suited to the HaPPY code and other quantum codes based on stabilizers.

In the case that we considered of an eternal black hole, a thermofield double state can be interpreted as the other side of a wormhole. Such a picture led to some new thoughts on the black hole information paradox, in particular to the ER=EPR conjecture \cite{Maldacena:2013xja}. In such a picture, the other side of the wormhole is assimilated to the early radiation of an evaporating black hole. It would be nice to see if our GNS technology can lead to a more rigorous understanding of such statements, and to what extent it can help solve problems related to entanglement monogamy like the firewall paradox.

Another important point is that our construction is state-dependent, in the sense that the Hilbert spaces we will get can be very different depending on the KMS state we choose. This is an intrinsically infinite-dimensional subtlety, as KMS states are always unique for finite-dimensional quantum systems. In the presence of more than one KMS state, i.e. of a broken symmetry, one can construct inequivalent Hilbert space representations which will satisfy bulk reconstruction. It would then be interesting to study thermal ensembles of such representations of the boundary algebra, and to study to what extent they can be assimilated to a superposition of different geometries. We expect this question to give insight into some ensemble interpretations of the black hole information paradox, such as the Engelhardt-Wall construction \cite{Engelhardt:2018kcs,Engelhardt:2017aux} or $\alpha$-bits \cite{Hayden:2018khn}. An interesting tool from algebraic quantum field theory to use in such a context is the theory of superselection sectors. We hope to return to this problem in future work.

Finally, even if we expect Hilbert spaces and modular operators to play a key role in a full theory of quantum gravity due to the state-dependence we have discussed, it would be nice to know if in the semiclassical regime, there is a way to equate bulk and boundary relative entropies directly at the level of the $C^*$-algebras of state-independent observables. Indeed, there exists a more general definition of relative entropy for states on $C^*$-algebras which does not require von Neumann algebras and modular theory, and therefore is completely state-independent. If $\varphi$ and $\psi$ are positive linear functionals on a $C^*$-algebra, one can define their relative entropy by:
\begin{align}
\mathcal{S}(\varphi,\psi)=\underset{n\in\mathbb{N}}{\sup{}}\;\underset{x}{\sup{}}\left\{\varphi(1)\log{n}-\int_{\frac{1}{n}}^\infty(\varphi(y(t)^\dagger y(t))+t^{-1}\psi(x(t)x(t)^\dagger))\frac{dt}{t}\right\},
\end{align}
where the supremum is taken over stepfunctions $x$ with values in the $C^*$-algebra with finite range, and where $y(t)=1-x(t)$. It would be interesting to know if the conservation of relative entropy with this definition between the bulk and the boundary can be proven directly from bulk reconstruction at the level of $C^*$-algebras, and even whether a converse statement could be true.

\section*{Acknowledgments}
The authors are grateful to Ivan Burbano and Matilde Marcolli for discussions and Temple He and Craig Lawrie for helpful comments on this paper. M.J.K. is supported by a Sherman Fairchild Postdoctoral Fellowship. This material is based upon work supported by the U.S. Department of Energy, Office of Science, Office of High Energy Physics, under Award Number DE-SC0011632. 
E.G. is funded by ENS Paris and would like to thank Matilde Marcolli for her guidance and constant support.
Research at Perimeter Institute is supported in part by the Government of Canada through the Department of Innovation, Science and Economic Development Canada and by the Province of Ontario through the Ministry of Colleges and Universities.

\end{document}